%% file: main_v2.tex
\documentclass[12pt]{article}
\usepackage[english]{babel}

\usepackage{fullpage}
\usepackage{cite}

\input{macros}

\usepackage{authblk}
\usepackage[dvipsnames]{xcolor}

\renewcommand{\EE}[1]{\underset{\scaleobj{.8}{#1}}{\mathds{E}\,}}

\allowdisplaybreaks 
\usepackage{graphicx,import,xcolor,transparent} 
\usepackage{centernot}
\newcommand{\rrangle}[0]{\rangle\!\rangle} 
\newcommand{\llangle}[0]{\langle\!\langle}

\begin{document}

\title{New approaches to almost i.i.d.\ information theory
}

\author[1]{Filippo Girardi\thanks{filippo.girardi@sns.it}}
\author[2]{Giacomo De Palma}
\author[1]{Ludovico Lami}

\affil[1]{Scuola Normale Superiore, Piazza dei Cavalieri 7, 56126 Pisa, Italy}
\affil[2]{Department of Mathematics, University of Bologna, Piazza di Porta San Donato 5, 40126 Bologna, Italy}

\date{}
\setcounter{Maxaffil}{0}
\renewcommand\Affilfont{\itshape\small}

\maketitle
\begin{abstract}
Independent and identically distributed (i.i.d.)\ states are ubiquitous in quantum information theory. However, in a practical setting, the i.i.d.\ assumption is too stringent, and possibly not realistic. A physically more compelling class of `almost i.i.d.'\ sources was recently proposed by~\href{https://arxiv.org/abs/2603.15792}{[Mazzola/Sutter/Renner, arXiv:2603.15792]}. In this paper, we introduce two alternative definitions of almost i.i.d.\ states, based on the normalised quantum Wasserstein distance and on the idea of looking at the average $k$-body marginal. We explore some basic properties of these notions and prove a strict hierarchical relation among them, with Mazzola et al.'s notion being the strictest, the one based on $k$-body marginals the loosest, and the one based on the quantum Wasserstein distance in between. Strict separation is established by means of explicit examples. 
\end{abstract}

\section{Introduction}\label{sec:intro}

The notion of independent and identically distributed (i.i.d.)\ sources plays a fundamental role in classical as well as quantum information theory. Conceptually, it allows us to extract the extensive behaviour of any correlation and information measure through the procedure of regularisation, analogous to taking the thermodynamic limit in statistical physics. This has allowed, historically, the definition and calculation of many important asymptotic quantifiers, such as the distillable entanglement~\cite{Bennett-distillation, Bennett-distillation-mixed}, or the classical~\cite{Holevo-S-W, H-Schumacher-Westmoreland} and quantum~\cite{Lloyd-S-D, L-Shor-D, L-S-Devetak} capacities. 

In spite of how natural it may be to consider them, \emph{exact} i.i.d.\ sources are not necessarily operationally justified. Physically, i.i.d.-ness stems from the assumption that whatever device is generating the source is memoryless, and that its internal state does not change over time. These assumptions will typically not holds, at least in a strict sense, for a real device. We instead expect a real device to retain a partial memory of its previous history, and that its internal state fluctuates in time: rather than being an exact i.i.d.\ source, it will only be an approximate one~\cite{Renner-talk-Cambridge, Mazzola_2026}. 

But how do we formalise mathematically this idea of `almost i.i.d.'\ source? A first answer to this question was provided by Renner, who introduced a notion of `almost product' (pure) states~\cite[Theorem~4.3.2]{RennerPhD} and employed it for cryptography. A generalisation of this notion to mixed states turned out to play a key role in Brand\~{a}o and Plenio's attempted proof of the generalised quantum Stein's lemma~\cite[Eq.~(66)-(67)]{Brandao2010}, and in one of the recent solutions of the same problem~\cite[Theorem~32]{GQSL} (see also~\cite[Corollary~25]{doubly-comp-classical}). A full discussion of this notion of almost i.i.d.-ness can be found in~\cite{Mazzola_2026}.

The purpose of this paper is to investigate two alternative formalisations of the concept of almost i.i.d.\ source. Mathematically, these are based on the notions of Wasserstein distance~\cite{De_Palma_2021} and of average local marginal state; we will therefore refer to the state modelling these sources as `Wasserstein almost i.i.d.'\ and `weakly almost i.i.d.'\ states. As we will show, these new notions can be thought of as relaxed versions of the one in~\cite{Mazzola_2026}. In fact, there is a strict hierarchy among the different notions: Mazzola--Sutter--Renner almost i.i.d.\ states are also Wasserstein almost i.i.d.,\ and Wasserstein almost i.i.d.\ states are also weakly almost i.i.d.\ Furthermore, we will construct explicit examples to demonstrate that all these inclusions are strict.

Besides illustrating alternative physical notions of almost i.i.d.-ness, our investigation relates the definition in~\cite{Mazzola_2026} to the concepts of Wasserstein distance and average $k$-body marginal. This connection opens the door to the use, in the study of almost i.i.d.\ states, of a variety of mathematical techniques known to apply to these latter concepts. We therefore hope that our contribution may help advance the programme laid out in~\cite{Renner-talk-Cambridge, Mazzola_2026}, whose goal is to ascertain whether physically relevant measures of quantum information are robust under almost i.i.d.\ perturbations.\bigskip

The remainder of the manuscript is organised as follows. First, we introduce our notation and discuss why the trace distance is not the optimal quantifier of almost i.i.d.-ness (Section~\ref{sec:1.1}). We then give a brief overview on the family of almost i.i.d.\ states introduced in~\cite{RennerPhD, Mazzola_2026} -- which we call Mazzola--Sutter--Renner almost i.i.d.\ states -- and we discuss a potential relaxation to a 'tailed' version (Section~\ref{sec:renner}). In Section~\ref{sec:almost} we introduce two new meaningful notions of almost i.i.d.\ states. More specifically, we initially follow the intuition of quantifying the impact of local defects by means of the quantum Wasserstein distance of order 1, which yields the novel definition of Wasserstein almost i.i.d.\ states (Section~\ref{sec:wasserstein}). Later on, we present the weakest possible definition of almost i.i.d.\ source, only relying on asymptotic properties of the average marginal of fixed size (Section~\ref{sec:weak}). We conclude by proving a strict inclusion relation among the three families of almost i.i.d.\ sources, which yields a hierarchical structure among the definitions discussed in this paper (Section~\ref{sec:inclusion}).

\subsection{Notation and motivation}\label{sec:1.1}
Let $\mathcal{H}$ be a finite dimensional Hilbert space, and let $\mathcal{D}(\mathcal{H})\subseteq \mathcal{L}(\mathcal{H})$ be the set of density matrices (or mixed states) on $\mathcal{H}$, including all linear, positive semi-definite operators \mbox{$\rho:\mathcal{H}\to \mathcal{H}$} with unit trace. When referring to vectors $\ket{\psi}$ in $\mathcal{H}$, for the sake of brevity we always implicitly assume the normalisation relation
$\braket{\psi|\psi}=1$, which implies that $\ket{\psi}$ is a pure state. Furthermore, the density matrix associated to the pure state $\ket{\psi}$ will be denoted as $\psi = \ketbra{\psi}$. Given $n\geq 2$, we say that a $n$-partite quantum state $\rho_n\in\mathcal{D}(\mathcal{H}^{\otimes n})$ is an \emph{independently and identically distributed} (i.i.d.)\ state if it has the form $\rho_n=\rho^{\otimes n}$ for some $\rho\in\mathcal{D}(\mathcal{H})$.

We denote by $\|\,\cdot\,\|_1$ the trace norm, defined by $\|A\|_1=\Tr\sqrt{A^\dagger A}$ for any linear operator $A:\mathcal{H}\to \mathcal{H}$. The trace norm induces the trace distance $d_{\rm tr}$, defined as
\bb
    d_{\rm tr}(\rho,\sigma)\coloneqq \frac 12 \|\rho-\sigma\|_1.
\ee
A fundamental operational interpretation of the trace distance is given by the Holevo--Helstrom theorem, which quantifies the one-shot success probability of a symmetric hypothesis testing problem with uniform prior in terms of the trace distance. More precisely, if a quantum system is prepared either in the state $\rho$ or in the state $\sigma$, with equal probabilities, then the optimal measurement to guess the state of the system has success probability
\bb
    p_{\rm succ}=\frac{1}{2}\big(1+d_{\rm tr}(\rho,\sigma)\big)
\ee
The trace distance is maximal, i.e.\ $d_{\rm tr}=1$, whenever the two considered states are orthogonal. 
Now, suppose $\eta\in\mathcal{D}(\mathcal{H})$ is orthogonal to $\rho\in\mathcal{D}(\mathcal{H})$. Then, for all $n\geq 1$, the states $\rho_n\coloneqq \rho^{\otimes n}$ and $\sigma_n\coloneqq\tau\otimes \rho^{\otimes (n-1)}$ are orthogonal. This means that their trace distance is maximal. The fact that an i.i.d.\ state $\rho_n$ may be maximally far away from one of its pointwise perturbations $\sigma_n$ is natural in the context of state discrimination, but, physically, it does not capture our intuition that $\rho_n$ and $\sigma_n$ are, in a sense, very similar states. 
Even more worryingly, according to the trace distance, the two states $\tau\otimes \rho^{\otimes (n-1)}$ and $\tau^{\otimes n}$ are equally distant from $\rho^{\otimes n}$, in spite of the fact that, while $\tau^{\otimes n}$ is globally different, $\tau\otimes \rho^{\otimes (n-1)}$ only differs from $\rho^{\otimes n}$ in a local sense. The purpose of this work is to try to make this intuition rigorous, and define mathematically what we mean by `globally' vs `locally' different.

We will need some more notation. Let $S_n$ be the symmetric group on the set $[n] = \{1,\ldots,n\}$, and let $\pi \in S_n$ be a permutation. We will denote by $U_\pi$ the unitary operator implementing $\pi$ on $\mathcal{H}^{\otimes n}$. We say that $\rho_n\in\mathcal{D}(\mathcal{H}^{\otimes n})$ is permutationally invariant if
\bb
    U_\pi\rho_nU_\pi^\dagger = \rho_n \qquad \forall \, \pi \in S_n.
\ee
Let $n\geq 1$ and let $I\subseteq[n]$ be a subset of $[n]$. We write $I^c\coloneq [n]\setminus I$ for the complement of $I$ with respect to $[n]$. Given any $\rho_n\in\mathcal{D}(\mathcal{H}^{\otimes n})$, we denote by $(\rho_n)_I$ the reduced density matrix
\bb
    (\rho_n)_I\coloneqq \Tr_{I^c}[\rho_n].
\ee
For any arbitrary $1\leq k\leq n$, the expectation value of a function $f$ of the uniformly chosen random subset $I\subseteq[n]$ of size $k$ will be written as
\bb
    \EE{\substack{I\subseteq [n]\\ |I|=k}}[f(I)]\coloneqq \frac{1}{\binom{n}{k}}\sum_{\substack{I\subseteq [n]\\ |I|=k}}f(I).
\ee
Finally, if $\mathcal{X}$ is a finite set, $\mathcal{P}(\mathcal X)$ will stand for the (convex) set of all probability distributions on $\mathcal{X}$. Given a distribution $P\in \mathcal{P}(\mathcal{X})$, we write the corresponding i.i.d.\ distribution on $\mathcal{X}^n$ as $P^{\times n}$; formally,
\bb
    P^{\times n}(x_1,\dots, x_n)\coloneqq P(x_1)\cdots P(x_n)
\ee
for all $x_1,\dots, x_n\in\mathcal{X}$.

\subsection{Mazzola--Sutter--Renner almost i.i.d.\ states} \label{sec:renner}

In his seminal research on the foundations of probability, Italian mathematician Bruno de Finetti pointed out a crucial link between symmetry and the i.i.d.\ nature of infinite sequences of random variables $X_1, \dots, X_N, \dots$~\cite{deFinetti_1,deFinetti_2}. Informally speaking, if the probability of the distribution of the sequence is invariant under permutations of the random variables -- the sequence is called \emph{exchangeable} in this case -- then it is possible to prove that the law of a finite subset $\{X_1,\dots, X_N\}$ is given by a convex combination of i.i.d.\ distributions. A quantitative statement for finite sequences of random variables $X_1\dots, X_N,\dots X_{N+k}$ was put forth by Diaconis and Freedman~\cite{Diaconis}. Suppose that such finite sequence is exchangeable, and that the random variable take value in a finite space $\mathcal{X}$. Then, the probability distribution $P_N$ of a subsequence of size $N$, say $\{X_1,\dots, X_N\}$, admits a representation as a convex combination of i.i.d.\ distributions $Q^{\times N}$, with $Q\in\mathcal{P}(\mathcal{X})$. More precisely, there exists a probability measure $\mu$ on $\mathcal{P}(\mathcal{X})$ such that
\bb\label{eq:deFinetti}
    \left\|\,P_N-\int_{\mathcal{P}(\mathcal{X})}Q^{\times N}d\mu( Q)\,\right\|_1\leq \frac{2N|\mathcal{X}|}{N+k}.
\ee
The quantum generalisation in the case of infinite exchangeable sequences of states appeared for the first time in~\cite{Caves2002}, and it was then generalised in~\cite{Konig2005, Christandl_2007} to sequences of finite size with an estimate analogous to~\eqref{eq:deFinetti}. A breakthrough in this line of research was the formulation of an exponentially tighter bound, introduced by Renner~\cite{RennerPhD, Renner2007}, known as \emph{exponentially de Finetti theorem}, which can be achieved when the convex combination of i.i.d.\ is replaced by a convex combination of almost i.i.d.\ states. More precisely, let $\mathcal{H}$ be a Hilbert space of dimension $d$, and let $\ket{\Psi_{N+k}}$ be a symmetric pure state on $\mathcal{H}^{\otimes N+k}$, namely $U_\pi \ket{\Psi_{N+k}}=\ket{\Psi_{N+k}}$ for all $\pi\in S_n$, Then~\cite[SI, Theorem~1]{Renner2007}
\bb\label{eq:Renato}
    \left\|\,\Tr_k\Psi_{N+k}-\int_{\mathcal{H}}\Phi_\theta\, d\mu( \theta)\,\right\|_1\leq 3k^de^{-\frac{k(r+1)}{N+k}},
\ee
where $\mu$ is a probability measure on pure states $\ket{\theta}$ of $\mathcal{H}$, and
\bb
    \ket{\theta} \in\mathcal{H}\quad \mapsto\quad \ket{\Phi_\theta}\in\mathcal{H}^{\otimes N}
\ee
is a function from pure states of $\mathcal{H}$ to pure states of $\mathcal{H}^{\otimes N}$, with $\ket{\Phi_\theta}$ being a first prototype of (symmetric) almost i.i.d.\ state:
\bb
    \ket{\Phi_\theta} = \frac{1}{n!}\sum_{\pi\in S_N}U_\pi \ket{\theta}^{\otimes N-r}\otimes \ket{\omega^{(r)}}
\ee
for some $\ket{\omega^{(r)}}\in \mathcal{H}^{\otimes r}$. The above result suggests that the above class of approximations of i.i.d.\ states might be the relevant one for physical applications, which justifies the following definition.

\begin{Def}[(Mazzola--Sutter--Renner (MSR) almost i.i.d.\ states)]\label{def:DeFinetti_almost_iid} 
Let $\rho \in \mathcal{D}(\mathcal{H}_A)$, and $n,r \in \mathbb{N}$ such that $r \leq n$. Let $\pazocal{V}^n_r(\mathcal{H}_{AE}, \ket{\psi})$ be the set of vectors
\bb
\pazocal{V}^n_r(\mathcal{H}_{AE}, \ket{\psi})\coloneqq\{U_\pi(\ket{\psi}^{\otimes n-r} \otimes \ket{\omega^{(r)}}): \pi \in S_n, \ket{\omega^{(r)}} \in \mathcal{H}^{\otimes r}_{AE} \}.
\ee
Then, $\rho_n \in \mathcal{D}(\mathcal{H}_A^{\otimes n})$ is called a $\binom{n}{r}$-\emph{almost i.i.d.\ state in} $\rho $ if there exists a purification $\ket{\psi_\rho}_{AE}$ of $\rho$ and an extension $\rho_n^{A^n E^n}$ of $\rho_n^{A^n}$ such that
\begin{enumerate}[(i)]
\item $\rho_n^{A^n E^n}$ is permutation-invariant with respect to $(A_i,E_i) \leftrightarrow (A_j,E_j)$;
\item $\supp(\rho_n^{A^n E^n}) \subseteq  \mathrm{span}\,\pazocal{V}^n_r(\mathcal{H}_{AE},\ket{\psi_\rho}_{AE})$.
\end{enumerate}
\end{Def}

 The previous definition naturally extends to sequences of states $(\rho_n)_n$ as follows.

\begin{Def} 
    Let $\rho\in\mathcal{D}(\mathcal{H})$ be a quantum state on a Hilbert space $\mathcal{H}$. We say that a sequence $(\rho_n)_n$ of states $\rho_n\in\mathcal{D}(\mathcal{H}^{\otimes n})$ is a \emph{MSR almost i.i.d.\ source along $\rho$} if, for $r_n=o(n)$, the state $\rho_n$ is $\binom{n}{r_n}$-almost i.i.d.
\end{Def}

Intuitively, we expect that physical properties of the i.i.d.\ source should be preserved if we replace it with an almost i.i.d.\ one in the above sense, provided that the number of `defects' $r$ is sublinear in $n$.
And indeed, the von Neumann entropy, the conditional entropy and the squashed entanglement have been shown to be robust in this sense~\cite{Mazzola_2026}. 
We start our discussion by suggesting  
two possible generalisations of Definition~\ref{def:DeFinetti_almost_iid}, which could turn out to be operationally meaningful:
\begin{itemize}
    \item first, we could relax the assumption that $\rho_n^{A^n}$ is permutation-invariant, as in general the noise acting on the ideal i.i.d.\ source could not have such a property;
    \item second, it might not be the case that the number of defects can be sharply constrained to be smaller than $r$: we may want to take into account some notion of 'tail' of the probability of observing a higher fraction of defects. More precisely, we introduce the following definition.
\end{itemize}

\begin{Def}[(Tailed MSR almost i.i.d.\ states)]
Let $\rho \in \mathcal{D}(\mathcal{H}_A)$, let $n\in \mathbb{N}$, and let $f:\mathbb N\to [0,+\infty]$ be a function. For $1\leq r\leq n$, let $\Pi_r$ be the orthogonal projector on the vector space spanned by
\bb
\pazocal{V}^n_r(\mathcal{H}_A, \ket{\psi}):=\{U_\pi(\ket{\psi}^{\otimes n-r} \otimes \ket{\omega^{(r)}}): \pi \in S_n, \ket{\omega^{(r)}} \in \mathcal{H}^{\otimes r} \}.
\ee
Then, $\rho_n \in \mathcal{D}(\mathcal{H}_A^{\otimes n})$ is called a $(n,f,\eta)$-\emph{tailed almost i.i.d.\ state in} $\rho $ if there exists a purification $\ket{\psi_\rho}_{AE}$ of $\rho$ and an extension $\rho_n^{A^n E^n}$ of $\rho_n^{A^n}$ such that
\bb\label{eq:tail}
    \sum_{r=1}^n f(r)\Tr\big[(\Pi_r-\Pi_{r-1})\rho_n^{A^nE^n}\big]\leq \eta,
\ee
where $\Pi_0\coloneqq0$.
\end{Def}

Let us briefly comment on a few examples.
\begin{itemize}
    \item $f(r)=+\infty\cdot \id_{r>r_0}$: for $\eta\neq +\infty$, we get standard $\binom n {r_0}$-almost i.i.d.\ states (without the permutational symmetry constraint).
    \item $f(r)=\id_{r>r_0}$: the set of $(n,f,\eta)$-almost i.i.d.\ states corresponds to the $\sqrt \eta$-ball in trace distance around the set of $\binom{n}{r_0}$-almost i.i.d.\ states. Indeed, suppose $\rho_n^{A^n}$ satisfies \eqref{eq:tail}, namely
    \bb
        \Tr\big[\Pi_{r_0}\rho_n^{A^nE^n}\big]\geq 1-\eta
    \ee
    Then, 
    \bb
        \frac{1}{2}\|\rho_{A^n}-\tilde\rho_{A^n}\|_1 &\leq \frac{1}{2}\|\rho_{A^nE^n}-\tilde\rho_{A^nE^n}\|_1 \leq \sqrt{\eta},
    \ee
    due to the gentle measurement lemma~\cite{Davies1969,VV1999} in the form of~\cite[Lemma~9.4.1]{MARK}.
     Conversely, suppose $\rho_A^n$ is $\eta$-close in trace distance to a $\binom n {r_0}$-almost i.i.d.\ state $\tilde{\rho}_{A^n}$. Then, take an extension\footnote{This follows from Uhlmann's theorem.} $\rho_n^{A^nE^n}$ of $\rho_n^{A^n}$ such that $\|\tilde\rho_n^{A^nE^n}-\rho_n^{A^nE^n}\|_1\leq \sqrt{\eta}$, where $\tilde\rho_n^{A^nE^n}$ is the extension of $\tilde\rho_n^{A^n}$ appearing in the definition of $\binom n {r_0}$-almost i.i.d.\ state. Then, writing $\rho_n^{A^nE^n}=\tilde\rho_n^{A^nE^n}+\Delta$ with $\|\Delta\|_1\leq \sqrt{\eta}$
    \bb
    \sum_{r=1}^n f(r)\Tr\big[(\Pi_r-\Pi_{r-1})\rho_n^{A^nE^n}\big]&=\Tr\big[(\id-\Pi_r)\rho_n^{A^nE^n}\big]\\
    &=\underbrace{\Tr\big[(\id-\Pi_r)\tilde\rho_n^{A^nE^n}\big]}_{=0}+\underbrace{\Tr\big[(\id-\Pi_r)\Delta\big]}_{\leq \|\Delta\|_1}\leq \sqrt{\eta}.
    \ee
    \item $f(r)=r/n$: with this function, $\eta$ intuitively represents the average fraction of defects.
    \item $f(r)=e^{\lambda r}$: choosing this function, a sequence of $(n,f,\eta)$-states $(\rho_n)_n$ with uniformly bounded $\eta<+\infty$ has the guarantee that the tail of the number of defects decays exponentially at a rate at least equal to $\lambda$.
\end{itemize}

\section{Two new families of almost i.i.d.\ states}\label{sec:almost}

\subsection{Beyond trace distance: the quantum Wasserstein distance of order~1}\label{sec:wasserstein}

The standard distinguishability measures in quantum information theory such as the trace distance, the fidelity, and the relative entropy are all unitarily invariant, and therefore assign maximal distance to any pair of orthogonal states. This symmetry becomes inadequate when one wishes to quantify robustness under local perturbations. For instance, in multipartite systems it is often desirable that the state $|0\rangle^{\otimes n}$ be regarded as much closer to $|1\rangle\otimes |0\rangle^{\otimes (n-1)}$ than to $|1\rangle^{\otimes n}$, reflecting the fact that the former differs only locally from the reference state. More generally, one seeks distances compatible with the geometry induced by the Hamming metric on product spaces and stable under local modifications. Such robustness is especially relevant in almost i.i.d.\ quantum information theory, where one studies states that deviate from tensor-product structure only on a small fraction of subsystems.

These considerations are closely related to continuity properties of entropic quantities. The von Neumann entropy is inherently insensitive to local perturbations: acting nontrivially on a single qubit can change the entropy by at most a constant independent of the total system size. However, this stability cannot be faithfully captured by unitarily invariant distinguishability measures, since a local operation may map a state to an orthogonal one, thereby producing maximal distance. This mismatch motivates the search for alternative metrics that incorporate an underlying notion of locality.

In the classical setting, the Wasserstein distances from optimal transport theory provide a natural framework for such purposes~\cite{villani2008optimal}. Given a metric space $(\mathcal X,D)$, the Wasserstein distance of order $1$ (also called $W_1$ distance, Monge--Kantorovich distance~\cite{monge1781memoire,kantorovich2006translocation} or earth mover's distance) between two probability distributions on $\mathcal X$ measures the minimal average transportation cost required to transform one distribution into the other, where the cost of moving a unit of mass from the point $x\in\mathcal X$ to the point $y\in\mathcal X$ is $D(x,y)$.

When the underlying space consists of strings over a finite alphabet, the natural choice of metric is the Hamming distance, which counts the number of positions in which two strings differ. The associated $W_1$ distance, known as Ornstein’s $\bar d$-distance~\cite{ornstein1973application}, has proved highly successful in ergodic theory and information theory, especially in contexts involving weak dependence, coding with memory, and rate distortion theory~\cite{gray2011entropy}. Its key feature is precisely its sensitivity to local perturbations: distributions supported on strings differing in only a few coordinates remain close in $W_1$ distance even if they are perfectly distinguishable in total variation distance.

These observations have motivated the development of quantum analogues of the $W_1$ distance, capable of capturing the locality structure of multipartite quantum systems. Such distance provides a natural framework for quantifying approximate tensor-product structures and for studying continuity phenomena in almost i.i.d.\ quantum information theory.

The quantum Wasserstein distance of order $1$, or $W_1$ distance, is a generalization of the Hamming distance from the set of the strings of $n$ symbols to the set of the states of an $n$-partite quantum system~\cite{De_Palma_2021,De_Palma_2024b}. At its heart there is the notion of neighbouring states: we say that the states $\rho$ and $\sigma$ of the $n$-partite quantum system $A_1\ldots A_n$ are \emph{neighboring} if there exists a subsystem $A_i$ such that $\mathrm{Tr}_{A_i}\rho = \mathrm{Tr}_{A_i}\sigma$, \emph{i.e.}, such that $\rho$ and $\sigma$ coincide if $A_i$ is discarded. One can then define the \emph{quantum $W_1$ norm} as the norm whose unit ball is the convex hull of the differences between neighboring states. The \emph{quantum $W_1$ distance} is the distance on the set of the states of $A_1\ldots A_n$ induced by the quantum $W_1$ norm. 

\begin{Def}[(Quantum $W_1$ distance~\cite{De_Palma_2021})]
For any $\rho,\,\sigma\in \mathcal{D}(\mathcal{H}_A^{\otimes n})$, we define
\bb
    \|\rho-\sigma\|_{W_1} \coloneqq \min \Bigg\{ \;\sum_{i=1}^n c_i\quad  &\text{s.t.}\quad  && c_i\geq 0, \qquad  \rho-\sigma=\sum_{i=1}^n c_i\left(\tau^{(i)}-\eta^{(i)}\right), \\
    & \text{with}&& \tau^{(i)},\eta^{(i)}\in \mathcal{D}(\mathcal{H}^{\otimes n}_A), \quad \Tr_{A_i}\tau^{(i)}=\Tr_{A_i}\eta^{(i)} &\!\!\Bigg\}.
\ee
\end{Def}

The $W_1$ distance admits a dual formulation in terms of a generalization of the Lipshitz constant to quantum observables:

\begin{Def}[(Quantum Lipshitz constant~\cite{De_Palma_2021})]
Let $X$ be an observable of the $n$-partite quantum system $A_1\ldots A_n$.
For any $i\in[n]$, we define the dependence of $X$ on $A_i$ as
\begin{equation}
    \partial_i X = 2\min\left\{\left\|X - X_{A_i^c}\right\| :\ X_{A_i^c}\textnormal{ does not act on }A_i\right\}\,.
\end{equation}
The quantum Lipshitz constant of $X$ is then
\begin{equation}
    \|X\|_L = \max_{i\in[n]}\partial_i X\,.
\end{equation}
\end{Def}

The $W_1$ distance between two states is then equal to the maximum difference between the expectation values of observables with unit Lipshitz constant:
\begin{prop}[{\cite{De_Palma_2021}}]
    Let $\rho$ and $\sigma$ be states of the $n$-partite quantum system $A_1\ldots A_n$.
    Then,
    \begin{equation}
        \left\|\rho-\sigma\right\|_{W_1} = \max_{\|X\|_L=1}\mathrm{Tr}\left[\left(\rho-\sigma\right)X\right]\,.
    \end{equation}
\end{prop}

The usefulness of the quantum $W_1$ distance comes from the property that it recognizes when two states differ only in a small fraction of the subsystems, while even a difference in a single subsystem is sufficient to make the states orthogonal and therefore maximally far with respect to any distance that is unitarily invariant, such as the trace distance, the fidelity or the relative entropy.
Therefore, the quantum $W_1$ distance has a very broad range of applications.
It has been employed in the context of quantum machine learning for quantum generative adversarial networks~\cite{Kiani_2022}, in the context of quantum computing to determine limitations of variational quantum algorithms~\cite{De_Palma_2023}, in quantum cryptography for quantum differential privacy~\cite{hirche2023quantumdifferentialprivacyinformation}, and in quantum statistical mechanics to prove the equivalence between the canonical and microcanonical ensembles~\cite{De_Palma_2022, De_Palma_2025}, for the theory of learning of many-body states~\cite{De_Palma_2024, Rouz__2024, Rouz__2024b} and to prove rapid thermalisation for geometrically local Hamiltonians~\cite{Bardet_2024, Kochanowski_2025, bakshi2025dobrushinconditionquantummarkov}.

Among the several known properties of the quantum $W_1$ distance, the main ones are the following. First, the $W_1$ distance between a generic state and a product state can be upper bounded by the relative entropy.

\begin{thm}[{\cite[Theorem 2]{De_Palma_2021}}]\label{thm:Marton}
    Let $\rho$ and $\omega$ be states of the $n$-partite quantum system $A_1\ldots A_n$, with $\omega$ a product state.
    Then,
    \begin{equation}
        \left\|\rho-\omega\right\|_{W_1} \le \sqrt{\frac{n}{2}\,D(\rho\|\omega)}\,.
    \end{equation}
\end{thm}

Second, the von Neumann entropy per subsystem is uniformly continuous with respect to the $W_1$ distance per subsystem.
\begin{thm}[{(Continuity of the von Neumann entropy~\cite[Theorem 9.1]{De_Palma_2023b})}]\label{thm:SW}
For any $n\in\mathbb{N}$ and any $\rho,\,\sigma\in\mathcal{D}\left(\mathcal{H}^{\otimes n}\right)$ we have
\begin{equation}
    \frac{1}{n}\left|S(\rho) - S(\sigma)\right| \le h_2\left(w\right) + w\log\left(\left(\dim\mathcal{H}\right)^2-1\right),\qquad \text{where}\qquad w\coloneqq \frac{\left\|\rho-\sigma\right\|_{W_1}}{n}
\end{equation}
and $h_2(x) = -x\log x - \left(1-x\right)\log\left(1-x\right)$ is the binary entropy function.
\end{thm}

Motivated by the above discussion, we can now introduce the following definition.

\begin{Def}[(Wasserstein almost i.i.d.\ states)] Let $\rho\in\mathcal{D}(\mathcal{H})$ be a quantum state on a (finite-dimensional) Hilbert space $\mathcal{H}$ let $n\geq 1$ and $\epsilon\geq 0$. We say that a state $\rho_n\in\mathcal{D}(\mathcal{H}^{\otimes n})$ is a $(n,\epsilon)$-\emph{almost i.i.d.\ state in} $\rho$ (according to the Wasserstein distance of order 1) if  
    \bb
        \frac 1 n \|\rho_n-\rho^{\otimes n}\|_{W_1}\leq \epsilon.
    \ee
\end{Def}

 We can extend the previous definition to sequences of states $(\rho_n)_n$ as follows.

\begin{Def}
    Let $\rho\in\mathcal{D}(\mathcal{H})$ be a quantum state on a  Hilbert space $\mathcal{H}$. We say that a sequence $(\rho_n)_n$ of states $\rho_n\in\mathcal{D}(\mathcal{H}^{\otimes n})$ is a \emph{Wasserstein almost i.i.d.\ source along $\rho$} if  
    \bb
        \lim_{n\to\infty}\frac 1 n \|\rho_n-\rho^{\otimes n}\|_{W_1}=0.
    \ee
\end{Def}

\subsection{Asymptotical i.i.d.-ness and weakly almost i.i.d.\ sources}\label{sec:weak}

Consider a source that produces an $n$-partite state $\rho_n$, which we want to employ for some task. How can we guarantee that the state is close to an i.i.d.\ one? We could carry out a full tomography procedure, but, by doing so, we would effectively destroy it. A better option is to sacrifice a few subsystems, chosen randomly, and implement a tomography protocol on those only -- if the source was actually i.i.d., then any $k$-body marginal is also i.i.d. However, this procedure can only guarantee that for an arbitrarily large \emph{but fixed} $k\in \N^+$, in the limit where $n\to\infty$ the $k$-body marginals of $\rho_n$, i.e.\ the reduced states of $\rho_n$ on the sub-systems of size $k$, will be close to $\rho^{\otimes k}$, at least \emph{on average}. This motivates us to give the following definition.

\begin{Def}[(Weakly almost i.i.d.\ source)]
    Let $\rho\in\mathcal{D}(\mathcal{H})$ be a quantum state on a (finite-dimensional) Hilbert space $\mathcal{H}$. We say that a sequence $(\rho_n)_n$ of states $\rho_n\in\mathcal{D}(\mathcal{H}^{\otimes n})$ is a \emph{weakly almost i.i.d.\ source along $\rho$} if  
    \bb\label{eq:def_weak}
        \limsup_{n\to \infty} \EE{\substack{I\subseteq [n]\\ |I|=k}} \big\|(\rho_n)_I - \rho^{\otimes I}\big\|_1 = 0\qquad \forall k \in \mathbb{N}_+ ,
    \ee
    where the random variable $I$ is uniformly distributed over the subsets of $[n]$ of size $k$.
\end{Def}

\begin{Def} Let $X_n\in \mathcal{L}(\mathcal{H}^{\otimes n})$. We define the \emph{local variation} norm $\|\,\cdot\,\|_{LV}$ as
\bb
    \|X_n\|_{LV}\coloneqq \frac 12\sum_{k=1}^n 2^{-k}\EE{\substack{I\subseteq [n],\\ |I|=k}}\|\Tr_{I^c}[X_n]\|_1,
\ee
where the random variable $I$ is uniformly distributed over the subsets of $[n]$ of size $k$.
\end{Def}

The local variation norm induces a distance $\|\rho_n-\sigma_n\|_{LV}$ which metrises the notion of weakly almost i.i.d.\ state. More precisely, the following holds.

\begin{prop}\label{prop:metrisation}
    The sequence $(\rho_n)_n$ of states $\rho_n\in\mathcal{D}(\mathcal{H}^{\otimes n})$ is a weakly almost i.i.d.\ source along $\rho$ if and only if the local variation distance with respect to the corresponding i.i.d.\ state asymptotically vanishes, namely
    \bb
        \lim_{n\to\infty}\|\rho_n-\rho^{\otimes n}\|_{LV}=0.
    \ee
\end{prop}

\begin{proof}
    Suppose that $\displaystyle{\lim_{n\to\infty}\|\rho_n-\rho^{\otimes n}\|_{LV}=0}$. Then, for any arbitrary $k \in \mathbb{N}_+$, we have
    \bb
        \limsup_{n\to \infty} \EE{\substack{I\subseteq [n],\\ |I|=k}}\big\|(\rho_n)_I - \rho^{\otimes I}\big\|_1 \leq \limsup_{n\to \infty}\big(2^{k+1} \|\rho_n-\rho^{\otimes n}\|_{LV}\big)=  0.
    \ee
    Conversely, suppose~\eqref{eq:def_weak} holds. Fix any arbitrary $\bar k \geq 1$; then, for $n\geq \bar k$,
    \bb
        \|\rho_n-\rho^{\otimes n}\|_{LV}\leq  \frac 12\sum_{k=1}^{\bar k} 2^{-k}\EE{\substack{I\subseteq [n],\\ |I|=k}} \big\|(\rho_n)_I-\rho^{\otimes k}\big\|_1+\frac 12\sum_{k=\bar k+1}^{n}2^{-k},
    \ee
    whence
    \bb
        \lim_{n\to \infty}\|\rho_n-\rho^{\otimes n}\|_{LV}\leq \frac 12\sum_{k=1}^{\bar k} 2^{-k} \lim_{n\to \infty}\EE{\substack{I\subseteq [n]\\ |I|=k}} \big\|(\rho_n)_I-\rho^{\otimes k}\big\|_1+2^{-\left(\bar k+1\right)}=\frac 1{2^{\bar k+1}}.
    \ee
    By arbitrariness of $\bar k$, we have concluded the proof.
\end{proof}

\section{The hierarchical structure of almost i.i.d.\ sources}\label{sec:inclusion}
\begin{figure}[t]
  \centering
  \def\svgwidth{0.75\linewidth}
  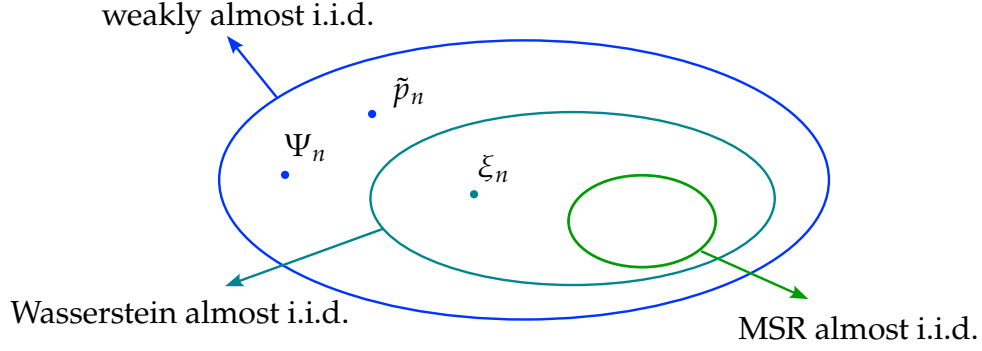
  \caption{A pictorial representation of the hierarchical relation between the various notions of almost i.i.d.-ness discussed in this paper: the Mazzola--Sutter--Renner notion is the strictest, the one derived from Wasserstein distance is the intermediate one, and, finally, weak almost i.i.d.-ness is the broadest. All inclusions are strict. 
  } 
  \label{fig:almost}
\end{figure}

A natural question that arises after the introduction of the new families of weakly and Wasserstein almost i.i.d.\ states and sources is how they mutually relate and if (and how) they are connected with the framework of Mazzola--Sutter--Renner almost i.i.d.\ states.

\begin{ex}
    For every $n\geq 1$, let
    \bb
        \rho_n\coloneqq\EE{\pi\in S_n}\big[U_\pi\big(\rho^{\otimes n-k_n}\otimes\omega_{k_n}\big)U^\dagger_\pi\big]
    \ee
    for a given fixed state $\rho\in\mathcal{D}(\mathcal{H})$ and for some defects $\omega_{k_n}\in\mathcal{D}\big(\mathcal{H}^{\otimes k_n}\big)$. This is a $\binom{n}{r}$-almost i.i.d.\ state (see \cite[Example~2.2]{Mazzola_2026}). Now, suppose $k_n=o(n)$ as $n\to\infty$, Then, as one can immediately verify through the following calculation, $(\rho_n)_n$ is a Wasserstein almost i.i.d.\ source:
    \bb
        \frac 1n \|\rho_n-\rho^{\otimes n}\|_{W_1} &\leq \frac 1n\EE{\pi\in S_n} \big\|U_\pi\big(\rho^{\otimes n-k_n}\otimes\omega_{k_n}\big)U^\dagger_\pi-\rho^{\otimes n}\big\|_{W_1} \\
        &=\frac 1n\big\|\rho^{\otimes n-k_n}\otimes\omega_{k_n}-\rho^{\otimes n}\big\|_{W_1}\\
        &=\frac 1n\big\|\omega_{k_n}-\rho^{\otimes k_n}\big\|_{W_1} \leq \frac{k_n}{n} \tendsn{} 0 .
    \ee
    Is it true that all Mazzola--Sutter--Renner almost i.i.d.\ sources are Wasserstein almost i.i.d.\ sources?
\end{ex}

The aim of this section is to prove a strict inclusion structure, as depicted in Figure~\ref{fig:almost}, yielding a sharp hierarchical structure among the three families of almost i.i.d.\ sources.

\subsection{Wasserstein almost i.i.d.\ sources are weakly almost i.i.d.\ sources}

We start by connecting the two notions of local variation distance and Wasserstein distance through the following lemma.

\begin{lemma}[(Bounding the local variation distance with Wasserstein)]\label{lem:WandLV} Let $\rho_n,\sigma_n\in\mathcal{D}(\mathcal{H}^{\otimes n})$. Then,
    \bb
       \|\rho_n - \sigma_n\|_{LV}\leq \frac 1n \|\rho_n-\sigma_n\|_{W_1}.
    \ee
\end{lemma}
\begin{proof}
    First, let us prove that
    \bb\label{eq:claim_w}
       \EE{\substack{I\subseteq [n],\\ |I|=k}}\frac{1}{2} \big\|(\rho_n)_I - \rho^{\otimes I}\big\|_1\leq \frac{k}{n}\|\rho_n-\rho^{\otimes n}\|_{W_1}.
    \ee
    Call $\|\rho_n-\sigma_n\|_{W_1}=w$. Then there exists a family of non-negative coefficients $(c_i)_{i=1,\dots,n}$ such that $\sum_{i=1}^nc_i=w$
    \bb
        \rho_n-\sigma_n=\sum_{i=1}^n c_i\left(\tau^{(i)}-\eta^{(i)}\right)
    \ee
    for some states $\tau^{(i)}$ and $\eta^{(i)}$ satisfying $\Tr_i\tau^{(i)}=\Tr_i\eta^{(i)}$. In particular,
    \bb
        (\rho_n)_I-(\sigma_n)_I=\sum_{i\in I} c_i\left(\tau^{(i)}-\eta^{(i)}\right),
    \ee
    whence
    \bb
        \|(\rho_n)_I-(\sigma_n)_I\|_1\leq \sum_{i\in I} c_i \big\|\tau^{(i)}-\eta^{(i)}\big\|_1 \leq 2\sum_{i\in I} c_i.
    \ee
    As a consequence, we can easily bound
    \bb
        \EE{\substack{I\subseteq [n]\\ |I|=k}}\frac{1}{2} \big\|(\rho_n)_I - \rho^{\otimes I}\big\|_1 &= \frac{1}{\binom{n}{k}}\sum_{\substack{I\subseteq [n],\\ |I|=k}} \frac 12\big\|(\rho_n)_I - \rho^{\otimes I}\big\|_1 \\
        &\leq \frac{1}{\binom{n}{k}}\sum_{\substack{I\subseteq [n],\\ |I|=k}} \sum_{i\in I} c_i
         = \frac{1}{\binom{n}{k}}\sum_{i\in [n]}\sum_{\substack{I\ni i,\\ |I|=k}}  c_i 
        = \frac{\binom{n-1}{k-1}}{\binom{n}{k}}\sum_{i\in [n]}c_i = \frac{k}{n}w.
    \ee
    This proves~\eqref{eq:claim_w} and immediately implies the bound
    \bb
       \|\rho_n - \sigma_n\|_{LV}\leq \left(\sum_{k=1}^n2^{-k}\frac{k}{n}\right)\|\rho_n-\sigma_n\|_{W_1}\leq \frac2n \|\rho_n-\sigma_n\|_{W_1},
    \ee
    which completes the proof.
\end{proof}
By combining Proposition~\ref{prop:metrisation} with Lemma~\ref{lem:WandLV}, we immediately get the following result.
\begin{cor}\label{cor:W_implies_almost}
    All Wasserstein almost i.i.d.\ sources are a weakly almost i.i.d.\ sources.
\end{cor}

\subsection{MSR almost i.i.d.\ sources are Wasserstein almost i.i.d.\ sources}

\begin{prop}[($W_1$ distance from a pure product state)]\label{prop:pure_Wasserstein}
For any $\phi_1, \dots, \phi_n \in \mathcal{D}(\mathcal{H})$ pure states and any $\rho_n \in \mathcal{D}(\mathcal{H^{\otimes n}})$ we have
\begin{equation}
    \frac{1}{n} \|\rho_n - \phi_1 \otimes \dots \otimes \phi_n\|_{W_1} \le \sqrt{\frac{h_2(p) + p \log(d-1)}{2}} + p,
\end{equation}
where $d$ is the dimension of each subsystem and
\begin{equation}
    p = \frac{1}{n} \sum_{i=1}^n (1 - \Tr[\rho_{n,i} \phi_i])
\end{equation}
is the average infidelity between $\rho_{n,i}$ and $\phi_i$.
\end{prop}

\begin{proof}
For any $i=1, \dots, n$, let $\rho_{n,i}\coloneqq\Tr_{[n]\setminus\{i\}}\rho_n$ and
\begin{equation}
    \omega_i \coloneqq \Tr[\rho_{n,i} \phi_i] \phi_i + (1 - \Tr[\rho_{n,i} \phi_i]) \frac{\id - \phi_i}{d-1}.
\end{equation}
From the triangle inequality, we have
\begin{equation}
    \frac{1}{n} \|\rho_n - \phi_1 \otimes \dots \otimes \phi_n\|_{W_1} \le \frac{1}{n} \|\rho_n - \omega_1 \otimes \dots \otimes \omega_n\|_{W_1} + \frac{1}{n} \|\omega_1 \otimes \dots \otimes \omega_n - \phi_1 \otimes \dots \otimes \phi_n\|_{W_1}.
\end{equation}
We can also write
\begin{align}
    \frac{1}{n} \|\rho_n - \omega_1 \otimes \dots \otimes \omega_n\|_{W_1} &\leqt{(a)} \sqrt{\frac{1}{2n} D(\rho_n || \omega_1 \otimes \dots \otimes \omega_n)} \nonumber \\
    &= \sqrt{\frac{1}{2n} \left( -S(\rho_n) - \sum_{i=1}^n \Tr[\rho_{n,i} \log \omega_i] \right)} \nonumber \\
    &\leqt{(b)} \sqrt{-\frac{1}{2n} \sum_{i=1}^n \Tr[\rho_{n,i} \log \omega_i]} \nonumber \\
    &= \sqrt{\frac{1}{2n} \sum_{i=1}^n \left( h_2(\Tr[\rho_n \phi_i]) + (1 - \Tr[\rho_n \phi_i]) \log(d-1) \right)} \nonumber \\
    &\leqt{(c)} \sqrt{\frac{h_2(p) + p \log(d-1)}{2}}
\end{align}
where~(a) follows from~\cite[Theorem~2]{De_Palma_2021}; (b)~follows observing that $S(\rho_n) \ge 0$; (c)~is a consequence of Jensen's inequality applied to the concave function $h_2$.
We have
\begin{equation}
    \frac{1}{n} \|\omega_1 \otimes \dots \otimes \omega_n - \phi_1 \otimes \dots \otimes \phi_n\|_{W_1} \eqt{(d)} \frac{1}{2n} \sum_{i=1}^n \|\omega_i - \phi_i\|_1 = \frac{1}{n} \sum_{i=1}^n (1 - \Tr[\rho_n \phi_i]) = p,
\end{equation}
where~(d) holds due to~\cite[Corollary 1]{De_Palma_2021}. The claim follows.
\end{proof}

\begin{cor} Let $\rho_n \in \mathcal{D}(\mathcal{H}_A^{\otimes n})$ be a $\binom{n}{r}$-\emph{almost i.i.d.\ state in} $\rho\in\mathcal{D}(\mathcal{H})$. Then $\rho_n $ is a $(n,\epsilon)$-almost i.i.d.\ state according to $W_1$, with
\bb
    \epsilon= \sqrt{\frac{h_2(r/n) + (r/n) \log(d-1)}{2}} + \frac{r}{n}.
\ee
In particular, if $(\rho_n)_n$ is a MSR almost i.i.d.\ source along $\rho$, then it is also a Wasserstein almost i.i.d.\ source along $\rho$.
\end{cor}

\begin{proof}
    Let $\ket{\psi_\rho}_{AE}$ be a purification of $\rho$ and $\rho_n^{A^n E^n}$ be an extension of $\rho_n^{A^n}$ according to Definition~\ref{def:DeFinetti_almost_iid}. In particular, by property (ii) in Definition~\ref{def:DeFinetti_almost_iid}, we immediately get
    \bb
    \sum_{i=1}^n\Tr\big[\rho_n^{A_iE_i} \psi_\rho\big]\geq n-r,
    \ee
    as one sees by observing that $\rho_n^{A^nE^n}$ is supported on span of the eigenvectors of the observable $\sum_{i=1}^n (\psi_\rho)_i \otimes \id^{\otimes [n]\setminus\{i\}}$ corresponding to eigenvalues at least $n-r$. Therefore, $\frac{1}{n} \sum_{i=1}^n \big(1 - \Tr\big[\rho_n^{A_iE_i} \psi_\rho\big]\big)\leq \frac{r}{n}$. By Proposition~\ref{prop:pure_Wasserstein} we conclude that
    \bb
        \frac{1}{n} \big\|\rho_n - \rho^{\otimes n}\big\|_{W_1} \leq \frac{1}{n} \big\|\rho_n^{A^nE^n} - \psi_\rho^{\otimes n}\big\|_{W_1} \leq \sqrt{\frac{h_2(r/n) + (r/n) \log(d-1)}{2}} + \frac{r}{n},
    \ee
    The fact that the right-hand side vanishes when $r$ grows sublinearly in $n$ concludes the proof.
\end{proof}

\subsection{Tightness of the hierarchy via counterexamples}

Now that we have discussed three different notions of almost i.i.d.-ness, it is natural to wonder whether they are really all distinct. In what follows, we answer this question in the affirmative by constructing explicit examples of weakly almost i.i.d.\ states that are not Wasserstein almost i.i.d., and of states belonging to this latter set that are not almost i.i.d.\ in the Mazzola--Sutter--Renner sense.

\subsubsection{Not all the weakly almost i.i.d.\ sources are almost i.i.d.\ Wasserstein sources}

Not all the weakly almost i.i.d.\ sources are almost i.i.d.\ according to the $W_1$ distance. Let us start from a classical counterexample. Define $\tilde p_n$ to be the probability distribution generated as follows:
\begin{itemize}
    \item the symbols $x_j$ with odd index $j\in [n]$ are drawn in an i.i.d.\ fashion according to some distribution $p\in \mathcal{P}(\XX)$ with positive entropy $S(p)>0$;
    \item the symbols with even $j$, instead, are set to be equal to the preceding odd-$j$ symbol.
\end{itemize}  
For some $k\in \N^+$, a random subset $I \subseteq [n]$ of cardinality $k$ will not contain any pair of consecutive indices $\{j,j+1\}$, with $j$ odd, with asymptotically unit probability $\frac{n-2}{n} \frac{n-4}{n} \ldots \frac{n-2k}{n}$. When no such pair is present, the probability distribution of the symbols in $I$ is exactly $p^{\otimes I}$. Therefore, 
\bb
\EE{\substack{I\subseteq [n]\\ |I|=k}} \left\| (p_n)_I - p^{\otimes I} \right\|_1 \leq 2 \left( 1 - \frac{n-2}{n} \frac{n-4}{n} \ldots \frac{n-2k}{n} \right) \tendsn{} 0\, .
\ee
However, it is not difficult to verify that $S(\tilde p_n) = \ceil{n/2} S(p)$, so that
\bb
    \lim_{n\to\infty}\frac{1}{n}S(\tilde p_n)=\frac{1}{2}S(p),
\ee
which would not be possible if $(\tilde p_n)_n$ were an almost i.i.d.\ source according to the $W_1$ distance, due to Theorem~\ref{thm:SW}. Note that, if we take a uniformly random subset of distinct indices $\{i_1,\dots, i_k\}\subseteq [n]$ of size $k=o(\sqrt{n})$, then with high probability no consecutive indices will be chosen. The reader may wonder whether the distinction between weakly almost i.i.d.\ sources and Wasserstein sources showed in the previous counterexamples stems from the sublinear growth of $k$ with respect to $n$. This is actually not true, as the quantum counterexample below shows.

\begin{prop}[\cite{Arnaud_2013}]\label{prop:ame}
For any sufficiently large $n$, there exists a pure state $\Psi_n$ of $n$ qubits such that $(\Psi_n)_I$ is maximally mixed for any subset $I\subset[n]$ of size $|I|\le0.189\,n$.
\end{prop}

On the one hand, thanks to the local indistinguishability from the maximally mixed state, the source $(\Psi_n)_n$ is weakly almost i.i.d.\ along the maximally mixed state.

\begin{lemma}
The source $(\Psi_n)_n$ defined in Proposition~\ref{prop:ame} is weakly almost i.i.d.\ along the maximally mixed state.
\end{lemma}

\begin{proof}
Let us fix $k\in\mathbb{N}$.
For any $n\ge \frac{k}{0.189}$ and any $I\subset[n]$ with $|I|=k$, we have that $(\Psi_n)_I$ is maximally mixed, therefore the source $(\Psi_n)_n$ is weakly almost i.i.d.\ along the maximally mixed state.
\end{proof}

On the other hand, due to the continuity of the von Neumann entropy with respect to the $W_1$ distance, the source $(\Psi_n)_n$ cannot be almost i.i.d.\ along the maximally mixed state with respect to the $W_1$ distance.
Indeed, the following more general result holds:

\begin{prop}\label{prop:cont_vN_W1}
Let the source $(\rho_n)_n$ be almost i.i.d.\ along the state $\rho\in\mathcal{D}(\mathcal{H})$ with respect to the $W_1$ distance.
Then,
\begin{equation}\label{eq:charact_W1}
    \lim_{n\to\infty}\frac{S(\rho_n)}{n} = S(\rho)\,.
\end{equation}
\end{prop}

\begin{proof}
For any $n\in\mathbb{N}$, let
\begin{equation}
    w_n = \frac{1}{n}\left\|\rho_n - \rho^{\otimes n}\right\|_{W_1}\,.
\end{equation}
Since the source $(\rho_n)_n$ is almost i.i.d.\ along $\rho$ with respect to the $W_1$ distance, we have
\begin{equation}
    \lim_{n\to\infty}w_n = 0\,.
\end{equation}
We then have from Theorem~\ref{thm:SW}
\begin{equation}
    \limsup_{n\to\infty}\left|\frac{S(\rho_n)}{n} - S(\rho)\right| \le \limsup_{n\to\infty}\left(h_2(w_n) + w_n\log\left(\left(\dim\mathcal{H}\right)^2-1\right)\right) = 0\,.
\end{equation}
The claim follows.
\end{proof}

\begin{cor}
No source of pure states, including the source defined in Proposition~\ref{prop:ame}, can be almost i.i.d.\ along the maximally mixed state with respect to the $W_1$ distance.
\end{cor}

\subsubsection{Not all the Wasserstein almost i.i.d.\ sources are MSR almost i.i.d.\ sources}
Let $\tau\coloneqq\id_2/2$ be the maximally mixed state on $\mathcal{H}\coloneqq\mathbb{C}^{2}$, and let $u$ be the uniform probability distribution on $\{0,1\}$. For all $n\geq 1$, we define $\xi_n\in\mathcal{H}^{\otimes n}$ as
\bb\label{eq:counterex}
    \xi_n\coloneqq\begin{cases}
        \displaystyle{\frac{1}{\big|T_u^{(n)}\big|}\sum_{x^n\in T_{u}^{(n)}}\ketbra{x^{(1)}\cdots x^{(n)}}} & n \text{ even}\\[1.7em]
        \xi_{n-1}\otimes\tau &  n \text{ odd}\\
    \end{cases} 
\ee
where $T_{u}^{(n)}$ is the type class containing all strings $x^n$ with an equal number of $0$'s and $1$'s; formally, $T_{u}^{(n)}\coloneqq\{ x^n\in \{0,1\}^n: N(0|x^n)=n/2\}$. 
We need two technical results, whose proof is deferred to the end of this section.

\begin{prop}\label{prop:quantitative_Wass}
    Let $n\geq 1$ be any integer, $t\in\mathcal{T}_{n}$ any type on an alphabet of cardinality $d\coloneqq |\XX|$, and $u_{t,n}$ the uniform distribution of sequences $x^n$ of type $t$. Then
    \bb
        \frac{1}{n}\|u_{t,n}-t^{\otimes n}\|_{W_1}\leq \sqrt{\frac{(d-1) \log (n+1)}{2n}} .
    \ee
\end{prop}

\begin{lemma}\label{lem:f} Let $n,r$ be two positive integers, with $n$ even, such that $r<n/2$; let $f:\{0,1\}^n\to\mathbb R$ be a function that can be written as
\begin{equation}\label{eq:fl}
f=\sum_{\ell} f_\ell^{(r)},
\end{equation}
where, for all $\ell$, the function $f_\ell^{(r)}$ depends at most on $r$ coordinates of the input string. 
If $f(x)=0$ whenever $\|x\|_1\neq n/2$, then $f\equiv 0$.
\end{lemma}

As an immediate consequence of Proposition~\ref{prop:quantitative_Wass}, we see that $(\xi_n)_{n}$ is a Wasserstein almost i.i.d.\ source along $\tau$. Now we leverage Lemma~\ref{lem:f} to show that $(\xi_n)_n$ cannot be a Mazzola--Sutter--Renner almost i.i.d.\ source along $\tau$. More precisely, the following holds.

\begin{prop} For any arbitrary $n\geq 2$ and $r\leq (n-1)/4$, the (classical) state $\xi_n$ defined in~\eqref{eq:counterex} is not $\binom{n}{r}$-almost i.i.d.\ in the sense of Definition~\ref{def:DeFinetti_almost_iid}. Therefore, the sequence $(\xi_n)_n$ is not a MSR almost i.i.d.\ source along $\tau$.
\end{prop}

\begin{proof}
Let $\xi_n^{A^nE^n}$ be an arbitrary extension of $\xi_n^{A^n}$ acting on the Hilbert space of the Hilbert--Schmidt operators on $\left(\mathbb{C}^2\right)^{\otimes n}$.
Let
\begin{equation}
    \xi^{A^nE^n}_n = \sum_i p_i\,\|X_i\rrangle\llangle X_i\|\,,\qquad \xi^{A^n}_n = \sum_i p_i\,X_i\,X_i^\dag\,,
\end{equation}
where we denote with $\|X\rrangle$ the linear operator $X$ meant as a vector in the Hilbert space of the Hilbert--Schmidt operators.
If $\xi_n^{A^n}$ is a $\binom{n}{r}$-almost i.i.d.\ state, then
\begin{itemize}
    \item it must have $\tau^{\otimes n}$ as a reference i.i.d.\ state, as the single qubit marginal converges to $\tau$ (cf. \cite[Proposition~2.6]{Mazzola_2026});
        \item we have, for all indices $i$, $\|X_i\rrangle\in {\rm span}\pazocal{V}^n_r(\mathcal{H}_A, \ket{\tau})$, where $\ket\tau$ is a purification of $\tau$.
\end{itemize}
Hence, if $\rho$ is $\binom{n}{r}$-almost i.i.d., there must exist an extension with
\begin{equation}
X_i \in \mathrm{span}\left\{U_\pi\left(|x\rangle\langle y|\otimes\id_2^{\otimes\left(n-r\right)}\right)U_\pi^\dagger : \pi\in S_n,\,x,\,y\in\{0,1\}^r\right\}\qquad\forall\;i\,.
\label{eq:X_i_span}
\end{equation}
Then, for $n$ even, let us define the function 
\bb
    f(x)\coloneqq \braket{x|\rho_n|x}= \sum_ip_i\braket{x|X_iX_i^\dagger|x},
\ee
which is a sum of terms that depend at most on $2r$ coordinates of $x\in\{0,1\}^n$, due to~\eqref{eq:X_i_span}. By~\eqref{eq:counterex}, for all $x$ such that $\|x\|_1\neq n/2$, we must have $f(x)=0$. Therefore, by Lemma~\ref{lem:f}, whenever $2r<n/2$ we get $f\equiv 0$, which contradicts the fact that $\rho_n\neq 0$. If $n=2k+1$ is odd, the same argument works with $f_0(x)\coloneqq \sum_ip_i\braket{x_0|X_iX_i^\dagger|x_0}$, where $x_0=(x,0)$, which is a function of an odd number of coordinates $x=(x_1, \dots, x_{2k})$ that must vanish when $\|x\|_1\neq k$, and with $f_1(x)\coloneqq \sum_ip_i\braket{x_1|X_iX_i^\dagger|x_1}$, where $x_1=(x,1)$. This concludes the proof.
\end{proof}

\begin{proof}[Proof of Proposition~\ref{prop:quantitative_Wass}.]
The classical case of Theorem~\ref{thm:Marton} gives us immediately
\begin{align}
\big\|u_{t,n} - t^{\otimes n}\big\|_{W_1}^2 &\leq \frac{n}{2} D(u_{t,n}\|t^{\otimes n}) \nonumber\\
&= \frac{n}{2} \left( - \log \big| T_u^{(n)}\big| - \sum_{i=1}^n \sum_x (u_{t,n})_i(x) \log t(x)\right) \nonumber\\
&= \frac{n}{2} \left( - \log \big| T_u^{(n)}\big| - n \sum_x t(x) \log t(x)\right) \\
&= \frac{n}{2} \left( n H(t) - \log \big| T_u^{(n)}\big|\right) \nonumber\\
&\leq \frac{n}{2} (d-1) \log(n+1) , \nonumber
\end{align}
where $(u_{t,n})_i$ denotes the marginal of $u_{t,n}$ on the $i^\text{th}$ symbol, which is clearly equal to $t$, and, in the last inequality, we employed the standard estimate $\big| T_u^{(n)}\big| \geq (n+1)^{d-1}\, 2^{-nH(t)}$ on the size of a type class~\cite[Lemma~2.3]{CSISZAR-KOERNER}.
\end{proof}

\begin{rem}
An alternative direct proof leverages the interpretation of $\|u_{t,n} - t^{\otimes n}\|_{W_1}$ as earth mover's distance: since most of the weight of $t^{\otimes n}$ is concentrated on types that are $\sqrt{n}$-close to $t$, redistributing the weight of $u_{t,n}$ to obtain $t^{\otimes n}$ involves moving a total (approximately unit) probability weight by Hamming distance $\sim \sqrt{n}$.
\end{rem}

\begin{proof}[Proof of Lemma~\ref{lem:f}.]
Let $x=(x_1,\dots,x_n)\in\{0,1\}^n$. If $g$ is a function depending only on the coordinates $\{x_i\}_{i\in S}$ for some $S\subseteq[n]$, then we can write
\bb
	g(x)=\sum_{S'\subseteq S}c_{S'}\prod_{i\in S'}x_i
\ee
for some $c_{S'}\in \mathbb{R}$. Therefore, using this property for each $f_\ell$ in~\eqref{eq:fl}, we get that $f$ is a multilinear polynomial of degree at most $r$.
\begin{equation}\label{eq:fc}
	f(x)=\sum_{\substack{S\subseteq[n]\\ |S|\le r}} c_S \prod_{i\in S} x_i
\end{equation}
Using a standard inclusion-exclusion approach (M\"obius inversion), we can retrieve the coefficients $c_S$ by considering the action of the function on indicator vectors $1_T\in\{0,1\}^n$ of the sets $T\subseteq S$:
\begin{equation}\label{eq:mob}
c_S=\sum_{T\subseteq S}(-)^{|S|-|T|}f(1_T)
\end{equation}
Eq.~\eqref{eq:mob} can be easily proved as follows. Let us notice that, for all $T'\subset S$, we have
\begin{equation}
	\sum_{T:\,T'\subseteq T\subseteq S}(-)^{|T|}=(-)^{|T'|}\sum_{T_0:\,T_0 \subseteq S\setminus T'}(-)^{|T_0|}=(-)^{|T'|}\prod_{i\in S\setminus T'}(1-1)=0.
\end{equation}
Hence, 
\bb
	\sum_{T\subseteq S}(-)^{|S|-|T|}f(1_T)&=\sum_{T\subseteq S}(-)^{|S|-|T|}\sum_{T'\subseteq T}c_{T'}\\
	&= 
	\sum_{T'\subseteq S}\sum_{T'\subseteq T\subseteq S}(-)^{|S|-|T|}c_{T'}\\
	&= 
	c_S+\sum_{T'\subset S}(-)^{|S|}\sum_{T'\subseteq T\subseteq S}(-)^{|T|}c_{T'}=c_S,
\ee
which proves~\eqref{eq:mob}.
Taking any $S\subseteq[n]$ with $|S|\leq r$, for all $T\subseteq S$ we have $\|1_T\|_1\leq r<n/2$, hence $f(1_T)=0$ by hypothesis. By~\eqref{eq:mob}, we conclude that $c_S=0$ in every term of the sum~\eqref{eq:fc}.
\end{proof}

\subsection{Asymptotic upper semi-continuity of the von Neumann entropy}

In Proposition \ref{prop:cont_vN_W1} we have showed that the normalised von Neumann entropy of a $W_1$ almost i.i.d.\ source is asymptotically continuous. In this section, we will prove that -- among all weakly almost i.i.d.\ sources -- such property is a characterising signature of Wasserstein almost i.i.d.\ sources. More precisely, we will prove that a weak almost i.i.d.\ source is a $W_1$ almost i.i.d.\ source if and only if \eqref{eq:charact_W1} holds (see Propositions \ref{prop:cont_vN_W1} and \ref{prop:if}). We start by mentioning a general result on the von Neumann entropy of weakly almost i.i.d.\ sources.

\begin{prop}\label{prop:usc}
    Let $(\rho_n)_n$ be a weakly almost i.i.d.\ source along $\rho$. Then, the normalised von Neumann entropy of $(\rho_n)_n$ is asymptotically upper semi-continuous, i.e.\
    \bb\label{eq:claim}
        \limsup_{n\to\infty}\frac{1}{n}S(\rho_n)\leq S(\rho).
    \ee
\end{prop}
\begin{proof}
    By the subadditivity of the von Neumann entropy, calling $\rho_{n,i}\coloneqq\Tr_{[n]\setminus \{i\}}\rho_n$, we have
    \bb
        \frac{1}{n}S(\rho_n)\leq \frac{1}{n}\sum_{i=1}^n S(\rho_{n,i})\leq S(\rho)+\frac{1}{n}\sum_{i=1}^n|S(\rho_{n,i})-S(\rho)|.
    \ee
Then, by the continuity of the von Neumann entropy with respect to the trace distance \cite{Fannes1973, Audenaert2007}, 
\bb\label{eq:above_limsup}
    \frac{1}{n}S(\rho_n)\leqt{(i)} S(\rho)+ \frac{1}{n}\sum_{i=1}^n \big(\epsilon_{n,i}\log(d-1)+h_2(\epsilon_{n,i})\big)\leqt{(ii)} S(\rho)+\epsilon_n\log (d-1)+h_2(\epsilon_n).
\ee
where in (i) we have introduced
\bb
    \epsilon_{n,i}\coloneqq \frac 12 \|\rho_{n,i}-\rho\|_1, \qquad \epsilon_n\coloneqq\frac{1}{n}\sum_{i=1}^n\epsilon_{n,i}=\frac{1}{2}\,\EE{\substack{I\subseteq [n]\\ |I|=1}}\|(\rho_n)_I-\rho^{\otimes I}\|_1,
\ee
and in (ii) we have used the concavity of the binary entropy. By the very definition of weakly almost i.i.d.\ source, we have $\displaystyle{\lim_{n\to\infty}\epsilon_n=0}$, so taking the limsup in \eqref{eq:above_limsup} yields \eqref{eq:claim}.
\end{proof}
Now we prove that any weakly almost i.i.d.\ source is not Wasserstein if and only if 
\bb
\displaystyle{\limsup_{n\to\infty}\frac{1}{n}S(\rho_n)< S(\rho)}.
\ee

\begin{prop}\label{prop:if}
    Let $(\rho_n)_n$ be a weakly almost i.i.d.\ source along $\rho$ such that
    \bb\label{eq:hp_ent}
        \lim_{n\to\infty}\frac{1}{n}S(\rho_n)= S(\rho).
    \ee
    Then $(\rho_n)_n$ is a $W_1$ almost i.i.d.\ source along $\rho$.
\end{prop}

\begin{proof}
    For $\delta\in (0,1)$, let us call $\rho_\delta\coloneqq (1-\delta)\rho+\delta\frac{\id} d$.
By the triangle inequality, we have
\begin{equation}
    \frac{1}{n} \|\rho_n - \rho^{\otimes n}\|_{W_1} \le \frac{1}{n} \|\rho_n - \rho_\delta^{\otimes n}\|_{W_1} + \frac{1}{n} \|\rho_\delta^{\otimes n}-\rho^{\otimes n}\|_{W_1}=\frac{1}{n} \|\rho_n - \rho_\delta^{\otimes n}\|_{W_1}+\frac 12\|\rho_\delta-\rho\|_1,
    \end{equation}
where the last identity follows from~\cite[Corollary 1]{De_Palma_2021}. Note that $\frac 12\|\rho_\delta-\rho\|_1\leq \delta$ and $\rho_\delta\geq \frac \delta d \id$. Now, calling $\rho_{n,i}\coloneqq\Tr_{[n]\setminus \{i\}}\rho_n$, by~\cite[Theorem~2]{De_Palma_2021}, we can upper bound
\bb\label{eq:bound_above}
    \frac{1}{n} \|\rho_n - \rho_\delta^{\otimes n}\|_{W_1} &\leq  \sqrt{\frac{1}{2n} D(\rho_n ||\rho_\delta^{\otimes n})} \\
    &= \sqrt{\frac{1}{2n} \left( -S(\rho_n) - \sum_{i=1}^n \Tr[\rho_{n,i} \log \rho_\delta] \right)} \\
    &\leqt{(i)} \sqrt{\frac{1}{2}\left(S(\rho_\delta)-\frac 1n S(\rho_n)\right)+\frac{1}{2}\log\left(\frac d \delta\right) \left(\|\rho-\rho_\delta\|_1+\frac 1n\sum_{i=1}^n\|\rho_{n,i}-\rho\|_1\right)} \\
    &\leq \sqrt{\frac{1}{2}\left(S(\rho_\delta)-\frac 1n S(\rho_n)\right)+\delta\log\left(\frac d \delta\right)+\frac{1}{2}\log\left(\frac d \delta\right)\EE{\substack{|I|\subseteq [n]\\ |I|=1}}\|(\rho_n)_I-\rho^{\otimes I}\|_1},
\ee
where~(i) simply follows from the inequality 
\bb
-\Tr[\rho_{n,i} \log \rho_\delta]&\leq -\Tr[\rho_\delta \log \rho_\delta]+\|\rho_{n,i}-\rho_\delta\|_1\|\log\rho_\delta\|_\infty\\
&\leq S(\rho_\delta)+\big(\|\rho_{n,i}-\rho\|_1+\|\rho-\rho_\delta\|_1\big)\,\big|\log \tfrac \delta d\big|.
\ee
Then, by the very definition of weakly almost i.i.d.\ source and by the assumption \eqref{eq:hp_ent}, we have
\bb
    \limsup_{n\to\infty}\frac{1}{n} \|\rho_n - \rho_\delta^{\otimes n}\|_{W_1} &\leq \sqrt{\frac{1}{2}\big(S(\rho_\delta)-S(\rho)\big)+\delta\log\left(\frac d \delta\right)}\\
    &\leq \sqrt{\frac{\delta}{2}\log (d-1)+\frac 12h_2(\delta)+\delta\log\left(\frac d \delta\right)}
\ee
where the last inequality follows from the continuity of the von Neumann entropy with respect to the trace distance \cite{Fannes1973, Audenaert2007}. By arbitrariness of $\delta$, we can take the limit $\delta\to 0^+$, which concludes the proof.
\end{proof}

\subsection{Collapse of the hierarchy for sources along pure states}
Leveraging the characterisation given in the previous section, any weakly almost i.i.d.\ source $(\rho_n)_n$ along a pure state $\gamma$ must also be a Wasserstein almost i.i.d.\ source. Indeed, by Proposition \ref{prop:usc},
\bb
	\limsup_{n\to\infty}\frac 1n S(\rho_n)\leq S(\gamma)=0,
\ee
which immediately implies, by positivity of the von Neumann entropy, that 
\bb
	\lim_{n\to\infty}\frac 1n S(\rho_n)= S(\gamma),
\ee
whence $(\rho_n)_n$ is a Wasserstein almost i.i.d.\ source along $\gamma$ by Proposition \ref{prop:if}. We are now going to prove another remarkable property of weakly almost i.i.d.\ sources $(\rho_n)_n$ along pure states: up to an asymptotically vanishing trace distance error, they correspond to MSR almost i.i.d.\ source along the same pure state (see Corollary \ref{cor:pure}). We need a preliminary statement.

\begin{prop}\label{prop:projector}
    Let $\gamma\in\mathcal{D}(\mathcal{H})$ be a pure state, let $\rho_n\in\mathcal{D}(\mathcal{H}^{\otimes n})$ and let
    \begin{equation}
    \epsilon = \frac{1}{n}\left\|\rho_n - \gamma^{\otimes n}\right\|_{W_1}.
    \end{equation}
    Let
    \begin{equation}
        H = \sum_{i=1}^n\left(\id-\gamma_i\right),
    \end{equation}
    where each $\gamma_i$ acts on the $i$-th subsystem as $\gamma$, and, for any $0\leq r\leq n$, we denote by $P_r$ be the orthogonal projector onto the subspace where $H\le r$.
    Then,
    \begin{equation}
        \frac{1}{2}\left\|\rho_n - P_r\,\rho_n\,P_r\right\|_1 \le \sqrt{\frac{n\,\epsilon}{r+1}}.
    \end{equation}
\end{prop}

\begin{proof}
Let $Q_r = \id - P_r$ be the orthogonal projector onto the subspace where $H\ge r+1$.
We have
\begin{equation}
    \left\|H\right\|_L \le 2\left\|\frac{\id}{2} - \gamma\right\| = 1,
\end{equation}
therefore
\begin{equation}\label{eq:stima}
\left(r+1\right)\mathrm{Tr}\left[\rho_n\,Q_r\right] \le \mathrm{Tr}\left[\rho_n\,H\right] = \mathrm{Tr}\left[\left(\rho_n - \gamma^{\otimes n}\right)H\right] \le \left\|\rho_n - \gamma^{\otimes n}\right\|_{W_1} = n\,\epsilon.
\end{equation}
Furthermore, writing $P_r\,\rho_n\,P_r=P_r\,\rho_n\,(\id-Q_r)$, we have
\bb
    \left\|\rho_n - P_r\,\rho_n\,P_r\right\|_1 = \left\|P_r\,\rho_n\,Q_r + Q_r\,\rho_n\right\|_1 &\le \left\|P_r\,\rho_n\,Q_r\right\|_1 + \left\|Q_r\,\rho_n\right\|_1 \\
    &\le \left\|\rho_n\,Q_r\right\|_1 + \left\|Q_r\,\rho_n\right\|_1 = 2\left\|Q_r\,\rho_n\right\|_1.
\ee
Let $\displaystyle{\rho_n = \sum_i p_i\,|\psi_i\rangle\langle\psi_i|}$
be a decomposition of $\rho_n$ as a convex combination of pure states.
Then, by concavity of the square root and by \eqref{eq:stima},
\bb
    \frac{1}{2}\left\|\rho_n - P_r\,\rho_n\,P_r\right\|_1&\leq \left\|Q_r\,\rho_n\right\|_1 \le \sum_i p_i\left\|Q_r|\psi_i\rangle\langle\psi_i|\right\|_1 = \sum_i p_i\sqrt{\langle\psi_i|Q_r|\psi_i\rangle} \\
    &\le \sqrt{\sum_i p_i\,\langle\psi_i|Q_r|\psi_i\rangle} = \sqrt{\mathrm{Tr}\left[\rho_n\,Q_r\right]}\le \sqrt{\frac{n\,\epsilon}{r+1}}\,.
\ee
The claim follows.
\end{proof}

\begin{cor}\label{cor:pure}
    Let $(\rho_n)_n$ be a Wasserstein (or, equivalently, weakly) almost i.i.d.\ source along a pure state $\gamma$. Then, there exists a MSR almost i.i.d.\ source $(\rho_n')_n$ along $\gamma$ such that
    \bb
        \lim_{n\to\infty}\|\rho_n-\rho_n'\|_1=0.
    \ee
\end{cor}

\begin{proof}
    Using the notation of Proposition \ref{prop:projector}, we have
    \bb
        1-\Tr[P_r\rho_n]\leq \|\rho-P_r\rho_nP_r\|_1\leq 2 \sqrt{\frac{n\,\epsilon_n}{r+1}}, \qquad \epsilon_n\coloneqq \frac{1}{n}\left\|\rho_n - \gamma^{\otimes n}\right\|_{W_1}.
    \ee
    Let $r=\lceil n\sqrt \epsilon_n\rceil =o(n)$.
    By definition of $P_r$, the sequence of states
    \bb
        \rho'_n\coloneqq \frac{P_r\rho_nP_r}{\Tr[P_r\rho_n]} \qquad n\geq 1,
    \ee
    is a MSR almost i.i.d.\ source along $\gamma$ with
    \bb
        \|\rho_n-\rho_n'\|_1&\leq \|\rho_n-P_r\rho_nP_r\|_1+\left\|P_r\rho_nP_r-\rho_n'\right\|_1\leq 2\sqrt{\frac{n\,\epsilon_n}{r+1}}+1-\Tr[P_r\rho_n]\\
        &\leq 4\sqrt{\frac{n\,\epsilon_n}{r+1}}\leq 4\epsilon_n^{1/4}=o(1)
    \ee
    This concludes the proof.
\end{proof}

\section{Conclusion}

In this paper, we have introduced two novel notions that model the intuitive concept of almost i.i.d.\ source in quantum information theory. These are based on the Wasserstein distance between quantum states and on the idea of looking at average $k$-body marginals, respectively. We have demonstrated a strict hierarchical relation between these classes of sources and a model of almost i.i.d.\ sources previously discussed by Mazzola et al.~\cite{Mazzola_2026}, highlighting the differences among these sets by means of explicit examples. We are confident that the tools we develop here will pave the way to a more complete understanding of approximate i.i.d.-ness in classical and quantum information theory, bringing up one step closer to a faithful, physically meaningful model of reality.

\subsection*{\textsc{Acknowledgements}} FG and LL acknowledge financial support from the European Union (ERC StG ETQO, Grant Agreement no.\ 101165230).
GDP has been supported by the UNA EUROPA SeedFunding project QUANTUMUnaE (CUP J37G25000380006).
GDP is a member of the ``Gruppo Nazionale per la Fisica Matematica (GNFM)'' of the ``Istituto Nazionale di Alta Matematica ``Francesco Severi'' (INdAM)''.

\bibliography{biblio}

\end{document}

%% file: macros.tex
\usepackage{newpxtext,newpxmath}

\let\coloneqq\relax

\usepackage[utf8]{inputenc}
\usepackage{amsthm}
\usepackage{amssymb}
\usepackage{amsmath}
\usepackage{bbold}
\usepackage{bbm}
\usepackage[pdftex, backref=page]{hyperref}
\usepackage{braket}
\usepackage{dsfont}
\usepackage{mathdots}
\usepackage{mathtools}
\usepackage{enumerate}
\usepackage[shortlabels]{enumitem}
\usepackage{csquotes}
\usepackage{stmaryrd}
\usepackage[cal=boondox]{mathalfa}
\usepackage{graphicx}
\usepackage{stackengine}
\usepackage{scalerel}
\usepackage{tensor}       
\usepackage{array}
\usepackage{makecell}
\newcolumntype{x}[1]{>{\centering\arraybackslash}p{#1}}
\usepackage{tikz}
\usepackage{pgfplots}
\usetikzlibrary{shapes.geometric, shapes.misc, positioning, arrows, arrows.meta, decorations.pathreplacing, decorations.pathmorphing, patterns, angles, quotes, calc}
\usepackage{booktabs}
\usepackage{xfrac}
\usepackage{siunitx}
\usepackage{centernot}
\usepackage{comment}
\usepackage{chngcntr}
\usepackage{caption}
\usepackage{subcaption}

\newtheorem{thm}{Theorem}
\newtheorem*{thm*}{Theorem}
\newtheorem{prop}[thm]{Proposition}
\newtheorem*{prop*}{Proposition}
\newtheorem{lemma}[thm]{Lemma}
\newtheorem*{lemma*}{Lemma}
\newtheorem{cor}[thm]{Corollary}
\newtheorem*{cor*}{Corollary}

\newtheorem*{cj*}{Conjecture}
\newtheorem{Def}[thm]{Definition}
\newtheorem*{Def*}{Definition}

\newtheorem*{question*}{Question}

\newtheorem*{problem*}{Problem}

\makeatletter
\def\thmhead@plain#1#2#3{%
  \thmname{#1}\thmnumber{\@ifnotempty{#1}{ }\@upn{#2}}%
  \thmnote{ {\the\thm@notefont#3}}}
\let\thmhead\thmhead@plain
\makeatother

\theoremstyle{definition}
\newtheorem{rem}[thm]{Remark}

\newtheorem{ex}[thm]{Example}

\newcommand{\bb}{\begin{equation}\begin{aligned}\hspace{0pt}}
\newcommand{\bbb}{\begin{equation*}\begin{aligned}}
\newcommand{\ee}{\end{aligned}\end{equation}}
\newcommand{\eee}{\end{aligned}\end{equation*}}
\newcommand*{\coloneqq}{\mathrel{\vcenter{\baselineskip0.5ex \lineskiplimit0pt \hbox{\scriptsize.}\hbox{\scriptsize.}}} =}

\newcommand\ceil[1]{\left\lceil#1\right\rceil}
\newcommand{\eqt}[1]{\stackrel{\mathclap{\scriptsize \mbox{#1}}}{=}}
\newcommand{\leqt}[1]{\stackrel{\mathclap{\scriptsize \mbox{#1}}}{\leq}}

\newcommand{\ketbra}[1]{\ket{#1}\!\!\bra{#1}}

\renewcommand{\epsilon}{\varepsilon}

\newcommand{\id}{\mathds{1}}

\newcommand{\N}{\mathds{N}}

\DeclareMathOperator{\Tr}{Tr}

\DeclareMathAlphabet{\pazocal}{OMS}{zplm}{m}{n}

\DeclareMathOperator{\supp}{supp}

\newcommand{\EE}{\pazocal{E}}

\newcommand{\XX}{\pazocal{X}}

\newcommand{\lsmatrix}{\left(\begin{smallmatrix}}
\newcommand{\rsmatrix}{\end{smallmatrix}\right)}

\stackMath
\newcommand\xxrightarrow[2][]{\mathrel{%
  \setbox2=\hbox{\stackon{\scriptstyle#1}{\scriptstyle#2}}%
  \stackunder[5pt]{%
    \xrightarrow{\makebox[\dimexpr\wd2\relax]{$\scriptstyle#2$}}%
  }{%
   \scriptstyle#1\,%
  }%
}}

\newcommand{\tendsn}[1]{\xxrightarrow[\! n\rightarrow \infty\!]{#1}}

\stackMath

\makeatletter
\newcommand*\rel@kern[1]{\kern#1\dimexpr\macc@kerna}
\newcommand*\widebar[1]{%
  \begingroup
  \def\mathaccent##1##2{%
    \rel@kern{0.8}%
    \overline{\rel@kern{-0.8}\macc@nucleus\rel@kern{0.2}}%
    \rel@kern{-0.2}%
  }%
  \macc@depth\@ne
  \let\math@bgroup\@empty \let\math@egroup\macc@set@skewchar
  \mathsurround\z@ \frozen@everymath{\mathgroup\macc@group\relax}%
  \macc@set@skewchar\relax
  \let\mathaccentV\macc@nested@a
  \macc@nested@a\relax111{#1}%
  \endgroup
}

\counterwithin*{equation}{part}
\counterwithin*{thm}{part}
\counterwithin*{figure}{part}

\tikzset{meter/.append style={draw, inner sep=10, rectangle, font=\vphantom{A}, minimum width=30, line width=.8, path picture={\draw[black] ([shift={(.1,.3)}]path picture bounding box.south west) to[bend left=50] ([shift={(-.1,.3)}]path picture bounding box.south east);\draw[black,-latex] ([shift={(0,.1)}]path picture bounding box.south) -- ([shift={(.3,-.1)}]path picture bounding box.north);}}}
\tikzset{roundnode/.append style={circle, draw=black, fill=gray!20, thick, minimum size=10mm}}
\tikzset{squarenode/.style={rectangle, draw=black, fill=none, thick, minimum size=10mm}}

\definecolor{Blues5seq1}{RGB}{239,243,255}
\definecolor{Blues5seq2}{RGB}{189,215,231}
\definecolor{Blues5seq3}{RGB}{107,174,214}
\definecolor{Blues5seq4}{RGB}{49,130,189}
\definecolor{Blues5seq5}{RGB}{8,81,156}

\definecolor{Greens5seq1}{RGB}{237,248,233}
\definecolor{Greens5seq2}{RGB}{186,228,179}
\definecolor{Greens5seq3}{RGB}{116,196,118}
\definecolor{Greens5seq4}{RGB}{49,163,84}
\definecolor{Greens5seq5}{RGB}{0,109,44}

\definecolor{Reds5seq1}{RGB}{254,229,217}
\definecolor{Reds5seq2}{RGB}{252,174,145}
\definecolor{Reds5seq3}{RGB}{251,106,74}
\definecolor{Reds5seq4}{RGB}{222,45,38}
\definecolor{Reds5seq5}{RGB}{165,15,21}

\allowdisplaybreaks

\usepackage[most,breakable]{tcolorbox}
	{\expandafter\ifstrequal\expandafter{#1}{orange}{\begin{tcolorbox}[colback=red!15,colframe=orange!15,breakable,enhanced]}{\begin{tcolorbox}[colback=Blues5seq1,colframe=Blues5seq5,breakable,enhanced]}}%
	{\end{tcolorbox}}

%% file: almost.pdf_tex
\begingroup%
  \makeatletter%
  \providecommand\color[2][]{%
    \errmessage{(Inkscape) Color is used for the text in Inkscape, but the package 'color.sty' is not loaded}%
    \renewcommand\color[2][]{}%
  }%
  \providecommand\transparent[1]{%
    \errmessage{(Inkscape) Transparency is used (non-zero) for the text in Inkscape, but the package 'transparent.sty' is not loaded}%
    \renewcommand\transparent[1]{}%
  }%
  \providecommand\rotatebox[2]{#2}%
  \newcommand*\fsize{\dimexpr\f@size pt\relax}%
  \newcommand*\lineheight[1]{\fontsize{\fsize}{#1\fsize}\selectfont}%
  \ifx\svgwidth\undefined%
    \setlength{\unitlength}{380.17526702bp}%
    \ifx\svgscale\undefined%
      \relax%
    \else%
      \setlength{\unitlength}{\unitlength * \real{\svgscale}}%
    \fi%
  \else%
    \setlength{\unitlength}{\svgwidth}%
  \fi%
  \global\let\svgwidth\undefined%
  \global\let\svgscale\undefined%
  \makeatother%
  \begin{picture}(1,0.35601552)%
    \lineheight{1}%
    \setlength\tabcolsep{0pt}%
    \put(0,0){\includegraphics[width=\unitlength,page=1]{almost.pdf}}%
    \put(0.0970736,0.33593541){\color[rgb]{0,0,0}\makebox(0,0)[lt]{\lineheight{1.25}\smash{\begin{tabular}[t]{l}weakly almost i.i.d.\end{tabular}}}}%
    \put(-0.00060754,0.01706579){\color[rgb]{0,0,0}\makebox(0,0)[lt]{\lineheight{1.25}\smash{\begin{tabular}[t]{l}Wasserstein almost i.i.d.\end{tabular}}}}%
    \put(0.7772657,0.00031845){\color[rgb]{0,0,0}\makebox(0,0)[lt]{\lineheight{1.25}\smash{\begin{tabular}[t]{l}MSR almost i.i.d.\end{tabular}}}}%
    \put(0,0){\includegraphics[width=\unitlength,page=2]{almost.pdf}}%
    \put(0.40798738,0.25676235){\color[rgb]{0,0,0}\makebox(0,0)[lt]{\lineheight{1.25}\smash{\begin{tabular}[t]{l}$\tilde p_n$\end{tabular}}}}%
    \put(0.49822647,0.17295137){\color[rgb]{0,0,0}\makebox(0,0)[lt]{\lineheight{1.25}\smash{\begin{tabular}[t]{l}$\xi_n$\end{tabular}}}}%
    \put(0.29213619,0.19565808){\color[rgb]{0,0,0}\makebox(0,0)[lt]{\lineheight{1.25}\smash{\begin{tabular}[t]{l}$\Psi_n$\end{tabular}}}}%
    \put(0,0){\includegraphics[width=\unitlength,page=3]{almost.pdf}}%
  \end{picture}%
\endgroup%

%% file: biblio.bib
@article{GQSL,
  author={Lami, L.},
  journal={IEEE Trans. Inf. Theory}, 
  title={A Solution of the Generalized Quantum {S}tein's Lemma}, 
  year={2025},
  volume={71},
  number={6},
  pages={4454--4484},
  doi={10.1109/TIT.2025.3543610}
}

@article{doubly-comp-classical,
  title={A doubly composite {C}hernoff--{S}tein lemma and its applications}, 
  author={Lami, L.},
  year={2025},
  journal={Preprint arXiv:2510.06342},
  doi={10.48550/arXiv.2510.06342}
}

@article{ornstein1973application,
  title={An application of ergodic theory to probability theory},
  author={Ornstein, D. S.},
  journal={Ann. Probab.},
  volume={1},
  number={1},
  pages={43--58},
  year={1973},
  publisher={JSTOR}
}

@book{gray2011entropy,
  title={Entropy and Information Theory},
  author={Gray, R.M.},
  isbn={9781441979704},
  series={Engineering},
  year={2011},
  publisher={Springer US}
}

@book{monge1781memoire,
  title={M{\'e}moire sur la th{\'e}orie des d{\'e}blais et des remblais},
  author={Monge, G.},
  year={1781},
  publisher={Imprimerie royale}
}

@article{kantorovich2006translocation,
  title={On the Translocation of Masses},
  author={Kantorovich, L. V.},
  journal={J. Math. Sci.},
  volume={133},
  number={4},
  year={2006},
  doi={10.1007/s10958-006-0049-2}
}

@book{villani2008optimal,
  title={Optimal Transport: Old and New},
  author={Villani, C.},
  isbn={9783540710509},
  lccn={2008932183},
  series={Grundlehren der mathematischen Wissenschaften},
  year={2008},
  publisher={Springer Berlin Heidelberg}
}

@article{Caves2002,
  title = {Unknown quantum states: The quantum de {F}inetti representation},
  volume = {43},
  ISSN = {1089-7658},
  url = {http://dx.doi.org/10.1063/1.1494475},
  DOI = {10.1063/1.1494475},
  number = {9},
  journal = {J. Math. Phys.},
  publisher = {AIP Publishing},
  author = {Caves, C. M. and Fuchs,  C. A. and Schack, R.},
  year = {2002},
  pages = {4537–4559}
}

@article{Konig2005,
  title = {A de {F}inetti representation for finite symmetric quantum states},
  volume = {46},
  ISSN = {1089-7658},
  url = {http://dx.doi.org/10.1063/1.2146188},
  DOI = {10.1063/1.2146188},
  number = {12},
  journal = {J. Math. Phys.},
  publisher = {AIP Publishing},
  author = {K\"{o}nig, R. and Renner, R.},
  year = {2005},
  month = Dec 
}

@article{Christandl_2007,
  title = {One-and-a-Half Quantum de {F}inetti Theorems},
  volume = {273},
  ISSN = {1432-0916},
  url = {http://dx.doi.org/10.1007/s00220-007-0189-3},
  DOI = {10.1007/s00220-007-0189-3},
  number = {2},
  journal = {Commun. Math. Phys.},
  publisher = {Springer Science and Business Media LLC},
  author = {Christandl, M. and K\"{o}nig, R. and Mitchison, G. and Renner, R.},
  year = {2007},
  month = Mar,
  pages = {473–498}
}

@article{Renner2007,
  title = {Symmetry of large physical systems implies independence of subsystems},
  volume = {3},
  ISSN = {1745-2481},
  url = {http://dx.doi.org/10.1038/nphys684},
  DOI = {10.1038/nphys684},
  number = {9},
  journal = {Nat. Phys.},
  publisher = {Springer Science and Business Media LLC},
  author = {Renner, R.},
  year = {2007},
  month = July,
  pages = {645–649}
}

@article{Diaconis,
 ISSN = {00911798, 2168894X},
 URL = {http://www.jstor.org/stable/2242823},
 author = {Diaconis, P. and Freedman, D.},
 journal = {Ann. Probab.},
 number = {4},
 pages = {745--764},
 publisher = {Institute of Mathematical Statistics},
 title = {Finite Exchangeable Sequences},
 urldate = {2026-05-10},
 volume = {8},
 year = {1980}
}

@article{deFinetti_1,
     author = {de Finetti, B.},
     title = {La pr\'evision : ses lois logiques, ses sources subjectives},
     journal = {Ann. Inst. Henri Poincar\'e},
     pages = {1--68},
     year = {1937},
     publisher = {Institut Henri Poincar\'e et Gauthier-Villars},
     volume = {7},
     number = {1},
     zbl = {0017.07602},
     language = {fr},
}

@article{deFinetti_2,
title = {Logical foundations and measurement of subjective probability},
journal = {Acta Psychol.},
volume = {34},
pages = {129-145},
year = {1970},
issn = {0001-6918},
doi = {https://doi.org/10.1016/0001-6918(70)90012-0},
url = {https://www.sciencedirect.com/science/article/pii/0001691870900120},
author = {de Finetti, B.},
}

@unpublished{Renner-talk-Cambridge,
title= {{Almost-IID information theory}},
author = {Renner, R.},
year = {2024},
month = {Aug},
note = {Workshop `Bridging Quantum Information and Mathematical Physics'},
url = {https://felixleditzky.info/bqm/}
}

@article{Mazzola_2026,
      title={Almost-iid information theory}, 
      author={Mazzola, G. and Sutter, D. and Renner, D.},
      year={2026},
      journal={Preprint arXiv:2603.15792},
      url={https://arxiv.org/abs/2603.15792}, 
}

@article{De_Palma_2023b,
   title={The {W}asserstein Distance of Order 1 for Quantum Spin Systems on Infinite Lattices},
   volume={24},
   ISSN={1424-0661},
   url={http://dx.doi.org/10.1007/s00023-023-01340-y},
   DOI={10.1007/s00023-023-01340-y},
   number={12},
   journal={Ann. Henri Poincar\'{e}},
   publisher={Springer Science and Business Media LLC},
   author={De Palma, G. and Trevisan, D.},
   year={2023},
   month={June}, pages={4237–4282} }

@article{bakshi2025dobrushinconditionquantummarkov,
      title={A {D}obrushin condition for quantum {M}arkov chains: Rapid mixing and conditional mutual information at high temperature}, 
      author={Bakshi, A. and Liu, A. and Moitra, A. and Tang, E.},
      year={2025},
      journal={Preprint arXiv:2510.08542},
      url={https://arxiv.org/abs/2510.08542}
}

@inbook{De_Palma_2024b,
   title={Quantum Optimal Transport: Quantum Channels and Qubits},
   ISBN={9783031504662},
   ISSN={2947-9460},
   url={http://dx.doi.org/10.1007/978-3-031-50466-2_4},
   DOI={10.1007/978-3-031-50466-2_4},
   booktitle={Optimal Transport on Quantum Structures},
   publisher={Springer Nature Switzerland},
   author={De Palma, G. and Trevisan, D.},
   year={2024},
   pages={203–239} }

@article{De_Palma_2022,
   title={Quantum Concentration Inequalities},
   volume={23},
   ISSN={1424-0661},
   url={http://dx.doi.org/10.1007/s00023-022-01181-1},
   DOI={10.1007/s00023-022-01181-1},
   number={9},
   journal={Ann. Henri Poincaré},
   publisher={Springer Science and Business Media LLC},
   author={De Palma, G. and Rouzé, C.},
   year={2022},
   pages={3391–3429} }

@article{De_Palma_2025,
   title={Quantum Concentration Inequalities and Equivalence of the Thermodynamical Ensembles: An Optimal Mass Transport Approach},
   volume={192},
   ISSN={1572-9613},
   url={http://dx.doi.org/10.1007/s10955-025-03464-3},
   DOI={10.1007/s10955-025-03464-3},
   number={6},
   journal={J. Stat. Phys.},
   publisher={Springer Science and Business Media LLC},
   author={De Palma, G. and Pastorello, D.},
   year={2025},
   month=June }

@article{Bardet_2024,
   title={Entropy Decay for {D}avies Semigroups of a One Dimensional Quantum Lattice},
   volume={405},
   ISSN={1432-0916},
   url={http://dx.doi.org/10.1007/s00220-023-04869-5},
   DOI={10.1007/s00220-023-04869-5},
   number={2},
   journal={Commun. Math. Phys.},
   publisher={Springer Science and Business Media LLC},
   author={Bardet, I. and Capel, \'{A}. and Gao, L. and Lucia, A. and Pérez-García, D. and Rouzé, C.},
   year={2024}
}

@article{Kiani_2022,
   title={Learning quantum data with the quantum earth mover’s distance},
   volume={7},
   ISSN={2058-9565},
   url={http://dx.doi.org/10.1088/2058-9565/ac79c9},
   DOI={10.1088/2058-9565/ac79c9},
   number={4},
   journal={Quantum Science and Technology},
   publisher={IOP Publishing},
   author={Kiani, B. T. and De Palma, G. and Marvian, M. and Liu, Z.-W. and Lloyd, S.},
   year={2022},
   pages={045002} }

@article{Rouz__2024b,
   title={Learning quantum many-body systems from a few copies},
   volume={8},
   ISSN={2521-327X},
   url={http://dx.doi.org/10.22331/q-2024-04-30-1319},
   DOI={10.22331/q-2024-04-30-1319},
   journal={Quantum},
   publisher={Verein zur Forderung des Open Access Publizierens in den Quantenwissenschaften},
   author={Rouzé, C. and Stilck França, D.},
   year={2024},
   month={Apr}, pages={1319} }

@article{Kochanowski_2025,
   title={Rapid Thermalization of Dissipative Many-Body Dynamics of Commuting Hamiltonians},
   volume={406},
   ISSN={1432-0916},
   url={http://dx.doi.org/10.1007/s00220-025-05353-y},
   DOI={10.1007/s00220-025-05353-y},
   number={8},
   journal={Commun. Math. Phys.},
   publisher={Springer Science and Business Media LLC},
   author={Kochanowski, J. and Alhambra, \'{A}. M. and Capel, \'{A}. and Rouzé, C.},
   year={2025},
   month=July }

@article{Rouz__2024,
   title={Efficient learning of ground and thermal states within phases of matter},
   volume={15},
   ISSN={2041-1723},
   url={http://dx.doi.org/10.1038/s41467-024-51439-x},
   DOI={10.1038/s41467-024-51439-x},
   number={1},
   journal={Nat. Commun.},
   publisher={Springer Science and Business Media LLC},
   author={Rouzé, C. and Stilck França, D. and Onorati, E. and Watson, J. D.},
   year={2024}}

@article{hirche2023quantumdifferentialprivacyinformation,
      title={Quantum Differential Privacy: An Information Theory Perspective}, 
      author={Hirche, C. and Rouzé, C. and Stilck França, D.},
      year={2023},
      journal={Preprint arXiv:2202.10717},
      url={https://arxiv.org/abs/2202.10717}
}

@article{De_Palma_2023,
   title={Limitations of Variational Quantum Algorithms: A Quantum Optimal Transport Approach},
   volume={4},
   ISSN={2691-3399},
   url={http://dx.doi.org/10.1103/PRXQuantum.4.010309},
   DOI={10.1103/prxquantum.4.010309},
   number={1},
   journal={PRX Quantum},
   publisher={American Physical Society (APS)},
   author={De Palma, G. and Marvian, M. and Rouzé, C. and Stilck França, D.},
   year={2023},
   month=Jan }

@article{De_Palma_2024,
   title={Classical shadows meet quantum optimal mass transport},
   volume={65},
   ISSN={1089-7658},
   url={http://dx.doi.org/10.1063/5.0178897},
   DOI={10.1063/5.0178897},
   number={9},
   journal={J. Math. Phys.},
   publisher={AIP Publishing},
   author={De Palma, G. and Klein, T. and Pastorello, D.},
   year={2024},
   month=Sept }

@article{Arnaud_2013,
   title={Exploring pure quantum states with maximally mixed reductions},
   volume={87},
   ISSN={1094-1622},
   url={http://dx.doi.org/10.1103/PhysRevA.87.012319},
   DOI={10.1103/physreva.87.012319},
   number={1},
   journal={Phys. Rev. A},
   publisher={American Physical Society (APS)},
   author={Arnaud, L. and Cerf, N. J.},
   year={2013},
   month=Jan }

@article{De_Palma_2021,
   title={The Quantum {W}asserstein Distance of Order 1},
   volume={67},
   ISSN={1557-9654},
   url={http://dx.doi.org/10.1109/TIT.2021.3076442},
   DOI={10.1109/tit.2021.3076442},
   number={10},
   journal={IEEE Trans. Inf. Theory},
   publisher={Institute of Electrical and Electronics Engineers (IEEE)},
   author={De Palma, G. and Marvian, M. and Trevisan, D. and Lloyd, S.},
   year={2021},
   month=oct, pages={6627–6643} }

@book{CSISZAR-KOERNER,
  author = {Csisz\'{a}r, I. and K\"{o}rner, J.},
  title = {Information theory: coding theorems for discrete memoryless systems},
  series = {Probability and Mathematical Statistics},
  publisher = {Academic Press, Inc., New York-London},
  year = {1981}
}

@book{MARK,
  title={Quantum Information Theory},
  author={Wilde, M. M.},
  isbn={9781316813300},
  year={2017},
  publisher={Cambridge University Press},
  edition={2nd}
}

@article{VV1999,
  title={Coding theorem and strong converse for quantum channels},
  author={Winter, A.},
  journal={IEEE Trans. Inf. Theory},
  volume={45},
  number={7},
  pages={2481--2485},
  year={1999},
  publisher={IEEE}
}

@article{H-Schumacher-Westmoreland,
  title = {Sending classical information via noisy quantum channels},
  author = {Schumacher, B. and Westmoreland, M. D.},
  journal = {Phys. Rev. A},
  volume = {56},
  issue = {1},
  pages = {131--138},
  numpages = {0},
  year = {1997},
  publisher = {American Physical Society},
  doi = {10.1103/PhysRevA.56.131}
}

@article{Holevo-S-W,
  author={Holevo, A. S.},
  journal={IEEE Trans. Inf. Theory},
  title={The capacity of the quantum channel with general signal states},
  year={1998},
  volume={44},
  number={1},
  pages={269--273},
  doi={10.1109/18.651037},
  issn={1557-9654}
}

@article{Lloyd-S-D,
  title = {Capacity of the noisy quantum channel},
  author = {Lloyd, S.},
  journal = {Phys. Rev. A},
  volume = {55},
  issue = {3},
  pages = {1613--1622},
  numpages = {0},
  year = {1997},
  publisher = {American Physical Society},
  doi = {10.1103/PhysRevA.55.1613}
}

@unpublished{L-Shor-D,
  title = {Lecture notes},
  author = {Shor, P.},
  note = {{MSRI Workshop on Quantum Computation}, Available at \url{https://www.msri.org/workshops/203/schedules/1181/documents/3863/assets/34150}},
  year = {2002}
}

@article{L-S-Devetak,
  author={Devetak, I.},
  journal={IEEE Trans. Inf. Theory},
  title={The private classical capacity and quantum capacity of a quantum channel},
  year={2005},
  volume={51},
  number={1},
  pages={44--55},
  doi={10.1109/TIT.2004.839515},
  issn={1557-9654}
}

@article{Bennett-distillation,
  title = {Concentrating partial entanglement by local operations},
  author = {Bennett, C. H. and Bernstein, H. J. and Popescu, S. and Schumacher, B.},
  journal = {Phys. Rev. A},
  volume = {53},
  issue = {4},
  pages = {2046--2052},
  numpages = {0},
  year = {1996},
  publisher = {American Physical Society},
  doi = {10.1103/PhysRevA.53.2046}
}

@article{Bennett-distillation-mixed,
  title = {Purification of Noisy Entanglement and Faithful Teleportation via Noisy Channels},
  author = {Bennett, C. H. and Brassard, G. and Popescu, S. and Schumacher, B. and Smolin, J. A. and Wootters, W. K.},
  journal = {Phys. Rev. Lett.},
  volume = {76},
  issue = {5},
  pages = {722--725},
  numpages = {0},
  year = {1996},
  doi = {10.1103/PhysRevLett.76.722}
}

@article{Brandao2010,
  author="Brand{\~a}o, F. G. S. L. and Plenio, M. B.",
  title="A Generalization of Quantum {S}tein's Lemma",
  journal="Commun. Math. Phys.",
  year="2010",
  volume="295",
  number="3",
  pages="791--828",
  issn="1432-0916",
  doi="10.1007/s00220-010-1005-z"
}

@article{Davies1969,
author = "Davies, E. B.",
journal = "Commun. Math. Phys.",
number = "4",
pages = "277--304",
publisher = "Springer",
title = "Quantum stochastic processes",
volume = "15",
year = "1969"
}

@article{BBM,
  title = {Quantum cryptography without {B}ell's theorem},
  author = {Bennett, C. H. and Brassard, G. and Mermin, N. D.},
  journal = {Phys. Rev. Lett.},
  volume = {68},
  issue = {5},
  pages = {557--559},
  numpages = {0},
  year = {1992},
  publisher = {American Physical Society},
  doi = {10.1103/PhysRevLett.68.557},
  url = {https://link.aps.org/doi/10.1103/PhysRevLett.68.557}
}

@phdthesis{RennerPhD,
  author = {Renner, R.},
  title = {Security of quantum key distribution},
  school = {ETH Zurich},
  year = {2005},
  doi = {10.48550/arXiv.quant-ph/0512258},
  note = {Preprint arXiv:quant-ph/0512258}
}

@article{Fannes1973,
author="Fannes, M.",
title="A continuity property of the entropy density for spin lattice systems",
journal="Commun. Math. Phys.",
year="1973",
volume="31",
number="4",
pages="291--294",
issn="1432-0916",
doi="10.1007/BF01646490"
}

@article{Audenaert2007,
  author={Audenaert, K. M. R.},
  title={A sharp continuity estimate for the von {N}eumann entropy},
  journal={J. Phys. A},
  volume={40},
  number={28},
  pages={8127},
  year={2007}
}
